%% file: ms.tex
\documentclass[11pt]{article}

\usepackage{amsthm}
\usepackage{graphicx} 
\usepackage{array} 

\usepackage{amsmath, amssymb, amsfonts, verbatim}
\usepackage{hyphenat,epsfig,multirow}

\usepackage[usenames,dvipsnames]{xcolor}
\usepackage[ruled,linesnumbered]{algorithm2e}

\usepackage{tcolorbox}
\tcbuselibrary{skins,breakable}
\tcbset{enhanced jigsaw}

\usepackage[compact]{titlesec}

\definecolor{DarkRed}{rgb}{0.5,0.1,0.1}
\definecolor{DarkBlue}{rgb}{0.1,0.1,0.5}

\usepackage{hyperref}
\hypersetup{
colorlinks=true,
pdfnewwindow=true,
citecolor=ForestGreen,
linkcolor=DarkRed,
filecolor=DarkRed,
urlcolor=DarkBlue
}

\usepackage{bm}
\usepackage{url}
\usepackage{xspace}
\usepackage[mathscr]{euscript}

\usepackage{mdframed}

\usepackage[noend]{algpseudocode}
\makeatletter
\def\BState{\State\hskip-\ALG@thistlm}
\makeatother

\usepackage{cite}
\usepackage{enumitem}

\usepackage[margin=1in]{geometry}

\newtheorem{theorem}{Theorem}
\newtheorem{lemma}{Lemma}[section]
\newtheorem{proposition}[lemma]{Proposition}
\newtheorem{corollary}[theorem]{Corollary}
\newtheorem{claim}[lemma]{Claim}
\newtheorem{fact}[lemma]{Fact}

\theoremstyle{definition}
\newtheorem{remark}[lemma]{Remark}
\newtheorem{observation}[lemma]{Observation}

\newtheorem*{claim*}{Claim}
\newtheorem*{proposition*}{Proposition}
\newtheorem*{lemma*}{Lemma}
\newtheorem*{problem*}{Problem}

\newtheorem{mdresult}{Result}
\newenvironment{result}{\begin{mdframed}[backgroundcolor=lightgray!40,topline=false,rightline=false,leftline=false,bottomline=false,innertopmargin=2pt]\begin{mdresult}}{\end{mdresult}\end{mdframed}}

\newtheorem{mdinvariant}{Invariant}

\allowdisplaybreaks

\usepackage{tikz}
\usetikzlibrary{arrows}
\usetikzlibrary{arrows.meta}
\usetikzlibrary{shapes}
\usetikzlibrary{backgrounds}
\usetikzlibrary{positioning}
\usetikzlibrary{decorations.markings}
\usetikzlibrary{patterns}
\usetikzlibrary{calc}
\usetikzlibrary{fit}
\usetikzlibrary{snakes}

\usepackage{caption}
\usepackage{subcaption}

\newcommand{\II}{\mathbb{I}}
\newcommand{\HH}{\mathbb{H}}
\newcommand{\DD}{\mathbb{D}}

\newcommand{\mi}[2]{\ensuremath{\II(#1 \,; #2)}}
\newcommand{\en}[1]{\ensuremath{\HH(#1)}}
\newcommand{\kl}[2]{\ensuremath{\DD(#1~||~#2)}}

\renewcommand{\qed}{\nobreak \ifvmode \relax \else
      \ifdim\lastskip<1.5em \hskip-\lastskip
      \hskip1.5em plus0em minus0.5em \fi \nobreak
      \vrule height0.75em width0.5em depth0.25em\fi}

\newcommand{\REM}[1]{}
\input{macros}

\title{Distributed and Streaming Linear Programming in Low Dimensions}
\author{Sepehr Assadi\footnote{Department of Computer Science, Princeton University. Supported in part by the Simons foundation Algorithms and Geometry collaboration. Majority of the work done while the author was a graduate student at University of Pennsylvania. Email: \texttt{sassadi@princeton.edu}.} \and
Nikolai Karpov\footnote{Department of Computer Science, Indiana University Bloomington. Supported in part by NSF CCF-1525024, IIS-1633215 and CCF-1844234. Email: \texttt{\{nkarpov,qzhangcs\}@indiana.edu}.} \and 
Qin Zhang\footnotemark[2]
}
\date{}

\begin{document}
\maketitle

\begin{abstract}
We study linear programming and  general LP-type problems in several big data (streaming and distributed) models.  We mainly focus on low dimensional problems in which the number of constraints is much larger than the number of variables. Low dimensional 
LP-type problems appear frequently in various machine learning tasks such as robust regression, 
support vector machines, and core vector machines.  As supporting large-scale machine learning queries in database systems has become an important direction for 
database research, obtaining efficient algorithms for low dimensional LP-type problems on massive datasets is of great value. In this paper we give both upper and lower bounds for  LP-type problems in distributed and streaming models.  Our bounds are 
almost tight when the dimensionality of the problem is a fixed constant. 
\end{abstract}

\input{intro}

\input{prelim}

\input{algorithm}

\input{application}

\input{lb}

\input{lb-implications}

\section*{Acknowledgments} Qin Zhang would like to thank Yufei Tao for introducing the problem (as well as the two-curve intersection problem as a means toward proving a lower bound for linear programming).
\bibliographystyle{plain}
\bibliography{paper,general}


\end{document}

%% file: macros.tex
\newcommand{\itfacts}[1]{Fact~\ref{fact:it-facts}-(\ref{part:#1})\xspace}
\newcommand{\event}{\ensuremath{\mathcal{E}}}

\newcommand{\tvd}[2]{\ensuremath{\norm{#1 - #2}_{tvd}}}
\newcommand{\Ot}{\ensuremath{\widetilde{O}}}
\newcommand{\eps}{\ensuremath{\varepsilon}}
\newcommand{\Paren}[1]{\Big(#1\Big)}
\newcommand{\Bracket}[1]{\Big[#1\Big]}
\newcommand{\bracket}[1]{\left[#1\right]}
\newcommand{\paren}[1]{\ensuremath{\left(#1\right)}\xspace}
\newcommand{\card}[1]{\left\vert{#1}\right\vert}

\newcommand{\IR}{\ensuremath{\mathbb{R}}}
\newcommand{\IQ}{\ensuremath{\mathbb{Q}}}

\newcommand{\inner}[2]{\langle #1, #2 \rangle}

\newcommand{\expect}[1]{\Exp\bracket{#1}}

\newcommand{\set}[1]{\ensuremath{\left\{ #1 \right\}}}
\newcommand{\poly}{\mbox{\rm poly}}
\newcommand{\polylog}{\mbox{\rm  polylog}}

\DeclareMathOperator*{\Exp}{\ensuremath{{\mathbb{E}}}}
\DeclareMathOperator*{\Prob}{\ensuremath{\textnormal{\textbf{Pr}}}}
\renewcommand{\Pr}{\Prob}

\newcommand{\Ex}{\Exp}

\newcommand{\etal}{{\it et al.\,}}

\newcommand{\supp}[1]{\ensuremath{\textnormal{supp}(#1)}}

\newcommand{\Prot}{\ensuremath{\Pi}}
\newcommand{\prot}{\ensuremath{\pi}}

\newenvironment{tbox}{\begin{tcolorbox}[
		enlarge top by=5pt,
		enlarge bottom by=5pt,
		 breakable,
		 boxsep=0pt,
                  left=4pt,
                  right=4pt,
                  top=10pt,
                  arc=0pt,
                  boxrule=1pt,toprule=1pt,
                  colback=white
                  ]
	}
{\end{tcolorbox}}

\newcommand{\abs}[1]{\left| #1 \right|}
\newcommand{\norm}[1]{\left\lVert #1 \right\rVert}

\newcommand{\vc}{\lambda}

\newcommand{\dist}{\operatorname{dist}}

\newcommand{\A}{\mathcal{A}}
\newcommand{\X}{\mathcal{X}}
\newcommand{\Y}{\mathcal{Y}}
\newcommand{\N}{\mathcal{N}}

\newcommand{\T}{\mathcal{T}}
\renewcommand{\H}{\mathcal{H}}

\renewcommand{\S}{\mathcal{S}}

\renewcommand{\A}{\mathcal{A}}

\newcommand{\B}{\mathcal{B}}
\newcommand{\V}{\mathcal{V}}

\newcommand{\istar}{\ensuremath{i^{*}}}

\newcommand{\zstar}{\ensuremath{z^{\star}}}

\renewcommand{\dist}[1]{\ensuremath{\mathcal{D}_{#1}}}

\newcommand{\rv}[1]{\ensuremath{\mathsf{#1}}}

\newcommand{\rProt}{\rv{\Prot}}

\newcommand{\Line}{\ensuremath{\textnormal{\textsf{LineSegment}}}\xspace}
\newcommand{\Stepcurve}{\ensuremath{\textnormal{\textsf{StepCurve}}}\xspace}

\newcommand{\distribution}[1]{\ensuremath{\textnormal{dist}(#1)}}

\renewcommand{\dist}{\ensuremath{\mathcal{D}}}

\newcommand{\CC}[1]{\ensuremath{\textnormal{\textsf{CC}}(#1)}\xspace}
\newcommand{\CCr}[2]{\ensuremath{\textnormal{\textsf{CC}}^{#2}(#1)}\xspace}

\newcommand{\AugIndex}{\ensuremath{\textnormal{\textsf{Aug-Index}}}\xspace}

\newcommand{\EvenI}{\ensuremath{\textnormal{\textsf{EvenInstance}}\xspace}}
\newcommand{\OddI}{\ensuremath{\textnormal{\textsf{OddInstance}}}\xspace}

\newcommand{\Instance}{\ensuremath{\textnormal{\textsf{Instance}}}\xspace}

\newcommand{\rZ}{\rv{Z}}
\newcommand{\rA}{\rv{A}}
\newcommand{\rB}{\rv{B}}
\newcommand{\rC}{\rv{C}}
\newcommand{\rD}{\rv{D}}

\newcommand{\errs}{\textnormal{errs}\xspace}

%% file: intro.tex
\section{Introduction}
\label{sec:intro}

As machine learning becomes pervasive, how to effectively support machine learning tasks in database systems has become an imminent question.  
In a recent paper~\cite{MVPV17}, Makrynioti \etal observed that many machine learning problems can be expressed by {\em linear programs} (LP). They designed a level of abstraction called \emph{SolverBlox} on top of a declarative language \emph{LogiQL}\footnote{An extended version of Datalog~\cite{ACGKOPVW15}.} as a framework for expressing linear program formulations.  The query in the format of 
SolverBlox will then be translated to a format supported by an LP solver for computing the solution.  
In this paper we consider the algorithmic side of this research direction, that is, we focus on the design of efficient LP solvers for large-scale datasets. In particular, we propose algorithms for linear programming in three popular ``big data'' models, namely, the 
\emph{coordinator} model~\cite{PVZ12}, the  \emph{streaming} model~\cite{MP80,AMS99}, and \emph{massively parallel computation} (MPC)~\cite{KSV10,GSZ11,BKS17}. We also provide almost matching lower bounds when the dimensionality of the linear 
program is a fixed constant.  

In the rest of the introduction we will start with the definition of the problem and the description of the computation models, and then present our results and discuss previous work.

\paragraph{Problem Definition.}  The basic linear programming problem can be described as follows: we have a set of $d$ variables $(x_1,\ldots,x_d)$ and a set of $n$ linear constraints each of which (indexed by $j$) is in the form of $\sum_{i=1}^d a_i^j x_i \le b^j$, 
where $a_i^j, b^j$ are coefficients and $d$ is the {\em dimension} of the problem. We also have an objective function $\sum_{i=1}^d c_i x_i$.  The goal is to find an assignment for variables that minimizes the objective function while satisfying all the constraints.  

Linear programming is a special case of a more general problem called \emph{LP-type problem}~\cite{MSW96}, which we will discuss in details in Section~\ref{sec:lp-type}. 
Besides linear programming, LP-type problems also include several other important problems in machine learning, such as \emph{Linear Support Vector Machines (SVM)}~\cite{BGV92}, which is widely used in classification
and regression analysis~\cite{GJ09,Burges98,CB99}), and \emph{Core Vector Machines}~\cite{TKC05}, which is used to speed up general SVM computation (or, Linear SVM augmented by the kernel trick~\cite{BGV92}).  We will give the formal definitions of these 
problems in Section~\ref{sec:app}. The algorithms we propose in this paper work for general LP-type problems. 

In this paper we are interested in the scenario when the dimension of the linear program (and LP-type problem in general) is small compared to the number of constraints. 
Various examples of linear programming and LP-type problems in machine learning are of this type: SVMs 
and regression problems (in particular, least absolute error regression that can be modeled by linear programming) 
are often over-constrained; in the problems of Chebyshev approximation and linear separability, the number of variables are typically small.

\paragraph{Computational Models.} We study linear programming and LP-type problems in the following big data models. 

\begin{itemize}[leftmargin=12pt]

\item \emph{The (multi-pass) streaming model.} \ \  In this model, we have a single machine which can make linear scans of the input data sequence. The task is to compute some function defined on the input data sequence. The goal is to minimize the {\em 
memory space usage} and the {\em number of passes} needed. This model captures data that cannot fit the memory, and on which sequential scan is much more efficient than random access.

\item \emph{The coordinator model.} \ \   In this model, we have $k$ sites and a central coordinator.  Each site is connected by a two-way communication channel with the coordinator.  The input is initially partitioned among the $k$ sites.  The task is for 
the sites and coordinator to jointly compute some function defined on the union of the $k$ datasets.  The computation proceeds in rounds: At the beginning of each round, the coordinator sends a message to each site, and then each site replies with a message back 
to the coordinator. At the end of the computation, the coordinator outputs the answer. The goal is to minimize the {\em total bits of the communication} and the {\em rounds of the computation}.
This model fits data that is inherently distributed or cannot fit the storage of a single machine

\item \emph{Massively parallel computation (MPC).} \ \ In this model, we  have $k$ machines interconnected in a network that allows communication between any pairs of machines. Similar to the coordinator model, the input is partitioned among the $k$ machines, 
and the task is for them to compute some function defined on the union of the $k$ datasets.  The computation is again in terms of rounds. At each round, the machines communicate with each other over the network by sending and receiving messages.  The 
message sent by a machine at each round is a function of its input data and all messages it has received in previous rounds.  Our goal is to minimize the {\em number of rounds of the computation}, and the \emph{maximum bits of information sent or received by a 
machine at any round} (often called the {\em load} in the literature).  MPC has already become the model of choice for studying parallel computation in computer clusters.
\end{itemize}


\emph{Description of the input.} Since we are dealing with low-dimensional problems, we assume that the memory on each site/machine in each model is at least proportional to $d$, the dimension of the problem, but is significantly smaller than $n$, the number of constraints. As a result, the input is presented by giving the constraints one by one to the algorithm in the  streaming model, or partitioning them across different sites/machines in the coordinator and MPC models. 

\subsection{Our Contributions and Related Work}\label{sec:results}

In the following, we present our results for linear programming in the three big data models described above, and postpone the specifics of their generalization to LP-type problems to later sections. 
Our main upper bound result is the following. 

\smallskip

\begin{result}\label{res:upper}
	We give the following \emph{polynomial time} algorithms for $d$-dimensional linear programming with $n$ constraints. For any integer $r \geq 1$ and parameter $\delta \in (0,1)$:
	\begin{itemize}[leftmargin=12pt]
		\item \emph{Streaming:} An $O(d \cdot r)$-pass streaming algorithm with $O(n^{1/r}) \cdot \poly(d,\log{n})$ space. 
		\item \emph{Coordinator:} An $O(d \cdot r)$-round distributed algorithm with $O(n^{1/r}+k) \cdot \poly(d,\log{n})$ total communication.
		\item \emph{MPC:} An $O(d/\delta^2)$-round  algorithm with $O(n^{\delta}) \cdot \poly(d,\log{n})$ load per machine. 
	\end{itemize}
	Our algorithms are randomized and output the correct answer with probability $1-1/n^{c}$ for any desired constant $c \geq 1$. 
\end{result}

\smallskip

\noindent
By Result~\ref{res:upper} for $r=\log{n}$ and $\delta=1/\sqrt{\log{n}}$, we obtain linear programming algorithms that use  $O(d \log{n})$ passes or rounds, and have space, communication, or load requirements in each model 
that is almost independent of the number of constraints. For low-dimensional instances, this results in a dramatic saving compared to direct implementations of 
standard LP algorithms in these models.   

Previously, Chan and Chen~\cite{CC07} proposed an $O(r^{d-1})$-pass streaming algorithm for linear programming that uses
$O(n^{1/r}) \cdot \poly(d,\log{n})$ space. Result~\ref{res:upper} improves upon this result by achieving an {exponentially} smaller pass-complexity in terms of $d$.

In the coordinator model, Daum{\'{e}} et al.~\cite{DPSV12} gave an algorithm 
using $O(r^{d + O(1)} \cdot k \cdot n^{1/r})$ communication based on an adaptation of the algorithm of~\cite{CC07}. The round-complexity and communication cost of this algorithm again depends {exponentially} on $d$. 

In the MPC model, very recently Tao~\cite{Tao18} gave a $d^{O(\log{(1/\delta)})}$-round MPC algorithm with load $O(n^{\delta})$ when $d =\polylog{(n)}$ (for any $\delta \in (0,1)$). 
This algorithm is then used as a building block for an 
interesting database application called {\em entity matching with linear classification}. The round complexity  of our MPC algorithm in Result~\ref{res:upper} improves that of \cite{Tao18} by an {exponential} factor.

To summarize, Result~\ref{res:upper} \emph{exponentially} improves upon the pass/round complexities of the state-of-the-art, while
using the same or smaller space, communication, or load, in the considered big data models. 

\smallskip

We complement our algorithms by giving almost \emph{tight} lower bounds for any fixed dimension (even $d = 2$) in the streaming and coordinator model.  
\smallskip


\begin{result}\label{res:lower}
	We give the following lower bounds for $2$-dimensional linear programming with $n$ constraints. For any integer $r \geq 1$:
	\begin{itemize}[leftmargin=12pt]
		\item \emph{Streaming:} Any $r$-pass algorithm requires $\Omega(n^{1/2r})$ space. 
		\item \emph{Coordinator:} Any $r$-round algorithm requires $\Omega(n^{1/2r})$ communication even when number of sites is only $k=2$. 
	\end{itemize}
	Our lower bounds hold even for randomized algorithms that output the correct answer with probability at least $2/3$. 
\end{result}

\smallskip


A few remarks about Result~\ref{res:lower}: Firstly, it is easy to see that linear programming in one dimension in the models we consider is a trivial task.
Result~\ref{res:lower} thus proves the lower bound for the smallest non-trivial dimension. We note that unlike Result~\ref{res:upper} that worked in all the three models, Result~\ref{res:lower} does \emph{not} prove any lower 
bound for MPC algorithms. Proving lower bounds for MPC algorithms is considered to be a challenging task as it has serious implications 
for long standing open problems in complexity theory~\cite{RoughgardenVW16}. Hence, no \emph{unconditional} lower bounds are known so far in the literature for \emph{any} MPC problem and Result~\ref{res:lower}
is of no exception. 

Prior to our work, Chan and Chen~\cite{CC07} gave a lower bound for $2$-dimensional linear programming for a \emph{restricted} family of \emph{deterministic} streaming algorithms in the \emph{decision tree} model (the only permitted operation of these 
streaming algorithms is testing the sign of a function evaluated at the coefficients of a subset of stored hyperplanes). Their lower bound states that this type of algorithms require $\Omega(n^{1/r})$ space to compute the solution in $r$ passes. Our lower bound
in Result~\ref{res:lower} is much stronger in that it proves a similar pass-space tradeoff for \emph{all} streaming algorithms (even randomized).  Finally, Guha and McGregor~\cite{GM08} showed that there is a fixed dimensional optimization problem for which any 
$r$-pass streaming algorithm requires $\Omega(n^{1/r})$ space. However, it is not clear how to adapt their proof to linear programming since their optimization problem involves quadratic constraints~\cite{McGregor18}.

\paragraph{Further Related Work.} Special cases of linear programming have been studied previously in the big data models. In particular, 
Ahn and Guha gave multi-pass streaming algorithms for $(1+\eps)$-approximation of \emph{packing} LPs~\cite{AG11} and Indyk~\etal~\cite{IMRUVY17} gave similar algorithms
for \emph{covering} LPs (see also~\cite{AssadiKL16}). These results focus on high-dimensional linear programs (non-constant $d$) and only packing/covering LPs, and are hence quite different from our approach in this paper.  


Unlike the case for big data models, low-dimensional linear programming has been studied extensively in the RAM model since the 1980s.  Megiddo~\cite{Megiddo84} gave an algorithm for $d$-dimensional linear programming with time complexity $O(2^{2^d} n)$, which is linear in terms of the number 
of constraints $n$.  This bound was consequently improved by a series of papers~\cite{Clarkson86,Dyer86,DF89,Kalai92,Clarkson95,ClarksonS89,MSW96,BCM99,Chan16}.



%% file: prelim.tex
\section{Preliminaries}\label{sec:prelim}

\paragraph{Notations.} For integers $1 \leq a \leq b$, we define $[a] := \set{1,\ldots,a}$, $[a : b] := \set{a,a+1,\ldots,b}$, and $(a:b] := [a:b] \setminus \set{a}$ (we define $[a: b)$ and $(a : b)$ analogously). 
We use capital letters for sets and random variables and calligraphic letters for set families. 
We use the notation ${\Ot}(f)$ to denote a function of the form $O(f \cdot \polylog{(f)})$.

Throughout the paper, we say an event happens ``with high probability'' if its probability can be lower bounded by $1-1/n^{c}$ for any desired constant $c \geq 1$ ($n$ is the number of constraints). 

We use the following standard variant of Chernoff bound.

\begin{proposition}[Chernoff bound]\label{prop:chernoff}
	Suppose $X_1,\ldots,X_t$ are $t$ independent random variables taking value in $[0,1]$ and $X := \sum_{i=1}^{t} X_i$. Then, for any $\eps > 0$, 
	\begin{align*}
		\Pr\Paren{\card{X - \expect{X}} > \eps \cdot \expect{X}} \leq 2 \cdot \exp\paren{\frac{-\eps^2 \cdot \expect{X}}{3}}.
	\end{align*}
\end{proposition}

\input{lp-type}

%% file: lp-type.tex
\subsection{LP-type Problems} 
\label{sec:lp-type}

We consider a generalization of linear programming referred to as LP-type problems\footnote{The class of LP-type problems is also known as \emph{abstract linear programming}~\cite{Bland78}.}.  An LP-type problem consists of a pair $(S, f)$, where $S$ is a finite set of elements, and $f : 2^S \to R$ is a set function with a range $R$ 
which is assumed to have a total order. The function $f$ satisfies two properties:
\begin{itemize}
\item {\em Monotonicity}:  for any two sets $X \subseteq Y \subseteq S$, $f(X) \le f(Y) \le f(S)$.

\item {\em Locality}: for any two sets $X \subseteq Y \subseteq S$, and any elements $e \in S$, if $f(X) = f(Y) = f(X \cup \{e\})$, then $f(Y) = f(Y \cup \{e\})$.
\end{itemize}
For an LP-type problem $(S, f)$, we call a set $B \subseteq S$ a {\em basis} of $S$ if $f(B) = f(S)$, and for all $B' \subset B$ we have $f(B') < f(B)$.
The goal is to compute a basis $B_S \subseteq S$ such that $f(B_S) = f(S)$. We say an element $e \in S$ {\em violates} $X \subseteq S$ if $f(X \cup \{e\}) > f(X)$.
It helps to think of an LP-type problem $(S,f)$ as an optimization problem in which elements of $S$ are the constraints, and $f(A)$ 
computes the best {\em feasible} solution on the set of constraints $A$. In the case when the optimal solution is not unique, we just break the tie arbitrarily.
Computing $f(B_S) = f(S)$ hence amounts to computing the optimal solution subject to all the constraints (we will make this connection 
explicit in the context of linear programming and other problems in Section~\ref{sec:app}).  

\paragraph{Combinatorial Dimension.} Note that an LP-type problem may have several bases which are of different sizes.  We define the {\em combinatorial dimension} of an LP-type problem to be the {\em maximum} cardinality of a basis for $S$, denoted by $\nu_{S, f}
$ ($\nu$ for short when $S$ and $f$ are clear from the context).

 

\subsection{$\eps$-Nets and VC Dimension}\label{sec:eps-net}


We now define another important notion that we use in designing our algorithms.  

\paragraph{VC Dimension.} A set-system is a tuple $(\H,U)$ consists of a universe $U$ and a set family $\H \subseteq 2^U$. 
Let $C \subseteq U$ be a set.  Define the intersection between a set family and a set to be the set family
\[
\H \cap C := \{H \cap C \ |\ H \in \H \}.
\]
We say that a set $C$ is {\em shattered} by $\H$ if $\H \cap C$ contains all the subsets of $C$, i.e.,
$
\abs{\H \cap C} = 2^{|C|}.
$
The \emph{VC dimension} of set-system $(\H,U)$, denoted by $\lambda_\H$ (or $\lambda$ for short when $\H$ is clear in the context), is then the cardinality of the \emph{largest} set $C$ that is shattered by $\H$.

\paragraph{$\eps$-Net.}  Given a set-system $(\X,U)$, and a weight function $w : \X \to \mathbb{R}$, for any $\Y \subseteq \X$,  let $w(\Y) := \sum_{Y \in \Y} w(Y)$.  We say a set $\N \subseteq \X$ is an $\eps$-net of $\X$ with respect to $w$ 
for a parameter $\eps \in (0,1)$, iff for any point $u \in U$ such that
$\sum_{X \in \X : u \notin X} w(X) \ge \eps \cdot  w(\X),
$
it holds that $\{X \in \N \ |\ u \notin X\} \neq \emptyset$.

The notion of $\eps$-net is well-studied in the literature (particularly in the computational geometry community~\cite{HW87, BG95, Mulmuley94}), and has been used in the algorithm design for many problems.  We use the following simple randomized construction of $\eps$-net for designing a distributed version of Clarkson's algorithm for LP-type problems. 

\begin{lemma}[\!\!\cite{HW87}]
\label{lem:eps-net}
For any set-system $(\X, U)$ of VC dimension $\lambda$, any weight function $w : \X \to \mathbb{R}$, and $\eps \in (0, 1)$, a set family $\N \subseteq \X$ obtained by randomly sampling
\begin{equation}
\label{eq:a-1}
m_{\eps, \vc, \delta} =   \max\left(\frac{8\vc}{\eps}\log{\frac{8\vc}{\eps}},\frac{4}{\eps}\log\frac{2}{\delta} \right)
\end{equation}
sets  with probability proportional to their weights is an $\eps$-net of $\X$ with probability at least $1 - \delta$.
\end{lemma}


%% file: algorithm.tex
\newcommand{\bit}[1]{\ensuremath{\textnormal{\texttt{\emph{bit}}}(#1)}\xspace}

\section{Algorithms}
\label{sec:algo}

In this section we present our algorithms for Result~\ref{res:upper}. We will work with a special class of LP-type problems that contains the most natural LP-type problems that we are aware of, including linear programming, Linear SVMs, and Core SVMs mentioned 
earlier. In particular, we require the LP-type problem $(\S,f)$ to satisfy the following properties: 
\begin{enumerate}[label=(P\arabic*)]
	\item\label{p:LP1} Each constraint $X \in \S$ is associated with a set of elements $S_X \subseteq R$ ($R$ is the range of $f$). 
	\item\label{p:LP2} For any $\A \subseteq \S$, $f(\A)$ is the minimal element of $\bigcap\limits_{X \in \A}X$. 
\end{enumerate}
It is useful to think of $R$ as the set of feasible solutions. For example, in the case of linear programming, $R = \IR^d$ with the natural ordering induced by scalar product with the vector $c$ in the objective function. 
Each constraint (inequality) $X \in \S$ corresponds to the subset of points $S_X$ which satisfy the constraint, and $f(\A)$ is equal to the point which satisfies all constraints in $\A$ and has a minimal scalar product with $c$. For convenience, we use $X$ and $S_X$ interchangeably. 

For this special class of LP-type problems, we define the VC dimension of the problem $(\S, f)$ as the VC dimension of the set system $(\S, R)$.

In the following, we first give a general meta-algorithm for solving LP-type problems with Properties~\ref{p:LP1} and~\ref{p:LP2}, and then show how to implement this meta-algorithm {efficiently} in each  model.

\subsection{The Meta Algorithm for LP-Type Problems}
\label{sec:meta}

Our meta-algorithm follows Clarkson's algorithm \cite{Clarkson95} for linear programming, but we use a different sampling procedure (by using $\eps$-net) which enables us to work with 
general LP-type problems with bounded VC dimension; it also significantly simplifies the analysis and facilitates the implementation of our algorithm in the big data models we consider. We further use a different weight increase rate after each iteration, which is essential for reducing the number of passes
in the streaming, and the number of rounds in the coordinator and MPC models.

The algorithm proceeds in iterations. We maintain a weight function $w : \S \rightarrow \IR$ throughout the algorithm which is initialized by setting $w(S) = 1$ for all $S \in \S$. 
In each iteration, we first sample a set family $\N$ of $m:=m_{\eps, \lambda_\S, \frac{2}{3}}$ sets from $\S$ with probability proportional to their weights so as to obtain an $\eps$-net $\N$ of $\S$ (according to Lemma~\ref{lem:eps-net}). We then compute a basis $\B$ of $\N$, and the set $\V$ of constraints which violate the basis $\B$.
If $w(\V) \leq \eps \cdot w(\S)$, then we say this iteration ``succeeds'', and update the weights of all sets
$S \in \V$ by setting $w(S) \leftarrow (n^{1/r}) \cdot w(S)$. Otherwise, we say this iteration ``fails'', and continue to the next one without modifying the weights. A pseudo-code  is provided in Algorithm~\ref{alg:meta}.

\begin{algorithm}[t]
	\caption{A Meta-Algorithm for LP-Type Problems}
	\label{alg:meta}

	\KwIn{An LP-type problem  $(\S, f)$ satisfying Properties~\ref{p:LP1} and~\ref{p:LP2} and integer $r \leq \ln{n}$.} 
	\KwOut{$f(\S)$.} 
	
	\medskip
	
	 Let $\eps := \frac{1}{10 \cdot \nu_{\S, f}  \cdot n ^{1/r}}$, and $\lambda$ as the VC dimension of the LP-type problem $(\S,f)$. \\
	 Set $w(S) = 1$ for every $S \in \S$. \\ 
	\Repeat{$\V = \emptyset$}{
		Sample a family $\N \subseteq \S$ of size $m := m_{\eps, \lambda, \frac{2}{3}}$ by picking each set in $\S$ with probability proportional to $w$ for the parameter $m_{\eps,\lambda,\frac{2}{3}}$ in Lemma~\ref{lem:eps-net}.  \label{ln:a-1} \\
		Compute a basis $\B$ of $\N$. \label{ln:a-2} \\
		Let $\V = \{S \in \S \ |\  f(\B \cup \{S\}) > f(\B)\}$ be the family of sets in $\S$ that violate $\B$. \label{ln:a-3} \\
		\If{$w(\V) \leq \eps \cdot w(\S)$ \label{ln:a-4}}{
			Set $w(S) = (n^{1/r}) \cdot w(S)$ for every set $S \in \V$. \label{ln:a-5}
		}
	}
	\KwRet{$f(\B)$}.
\end{algorithm}

In the following, we first establish the correctness of the meta-algorithm and then bound the number of iterations it needs. 

\begin{lemma}
\label{lem:main-correct}
When Algorithm~\ref{alg:meta} stops, it correctly computes $f(\S)$.
\end{lemma}
\begin{proof}
 At the end of the algorithm, we have $\V = \emptyset$. This means that for any $S \in \S \setminus \B$, we have $f(\B \cup \{S\}) = f(\B)$ by the monotonicity property of $f$.  By the locality property and induction we obtain that
$f(\B) = f(\B \cup (\S \backslash \B)) = f(\S)$, finalizing the proof. 
\end{proof}

We now bound the number of iterations. 
We say that an iteration of Algorithm~\ref{alg:meta} (at Lines~\ref{ln:a-1} to~\ref{ln:a-5}) is \emph{successful} iff $w(\V) \le \eps \cdot w(\S)$ in this iteration. 

\begin{claim}\label{clm:is-successful}
	Each iteration of Algorithm~\ref{alg:meta} is successful with probability at least $2/3$. 
\end{claim}
\begin{proof}
	Since the VC dimension of $(\S,R)$ is $\lambda$, by Lemma~\ref{lem:eps-net}, with probability at least $2/3$, the
	family $\N$ sampled in Line~\ref{ln:a-1} is an $\eps$-net for $(\S, R)$ with respect to the weight function $w$. In the following, we condition on this event. 
	
	Let $x := f(\B)$. By Property~\ref{p:LP2} of the LP-type problems we consider, we know that $x$ is the minimal element in the intersection of all sets in $\B$ according to the ordering of $R$. 
	For any set $S \in \S$ to violate $\B$, we need to have $x \notin S$; otherwise $f(\B \cup \set{S}) = x$ which is in contradiction with $f(\B \cup \set{S}) > f(\B)$. 
	Recall that $\V$ is the family of all sets in $\S$ that violate $\B$. Suppose
	towards a contradiction that $w(\V) > \eps \cdot w(\S)$. Since none of the sets in $\V$ contain $x$, and $\N$ is an $\eps$-net, by definition there is a set $S' \in \N$ where $S'$ does not contain $x$. 
	But this is in contradiction with $\B$ being a basis. To see this, if $f(\B) = f(\N)$, then $x$ belongs to all sets in $\N$, and consequently it should also be in $S'$. We thus have $w(\V) \leq \eps \cdot w(\S)$, finalizing the proof. 
	\end{proof}

\begin{lemma}
\label{lem:main-iteration}
The number of iterations in Algorithm~\ref{alg:meta} is $O(\nu \cdot r)$ with probability at least $1 - e^{-\Omega(\nu \cdot r)}$, where $\nu$ denotes the combinatorial dimension of $(\S,f)$.
\end{lemma}

\begin{proof}

Recall that the weight function $w(\cdot)$ is updated only when an iteration is successful, and each iteration succeeds with probability at least $2/3$ by Claim~\ref{clm:is-successful}. By Chernoff bound (Proposition~\ref{prop:chernoff}), we have that if the algorithm terminates in $t$ iterations, then with probability at least $1 - e^{-\Omega(t)}$, at least $t/2$ of these iterations are successful.

We now focus on successful iterations. Let $w_i(\cdot)$ be the weight function $w(\cdot)$ after the $i$-th successful iteration. Initially, for any $S \in \S$ we have $w_0(S) = 1$ (and thus $w_0(\S) = n$).  
We claim that for any integer $t \ge 1$, if Algorithm~\ref{alg:meta} reaches the $t$-th successful iteration, then
\begin{equation}
\label{eq:a-1}
n^{t/\nu r} \le w_{t}(\S) \le e^{t/10\nu} \cdot n.
\end{equation}
We establish Eq~(\ref{eq:a-1}) in the following two claims.

\begin{claim}
\label{cla:upper}
For any integer $t \ge 1$, we have $n^{t/\nu r} \le w_{t}(\S)$.
\end{claim}
\begin{proof}
Fix an arbitrary basis $\B^* = \set{B_1,\ldots,B_k}$ of $\S$ for some $k \leq \nu$ (recall that by definition, $\nu$ is size of the largest basis). 
Since $\B^* \subseteq \S$, we have $w_t(\B^*) \le w_t(\S)$ for any $t > 0$. We thus only need to show $n^{t/\nu r} \le w_{t}(\B^*)$.  

The first observation is that in any iteration, if $\V \neq \emptyset$ then we must have $\V \cap \B^* \neq \emptyset$.  Indeed, if $\V \cap \B^* = \emptyset$, then 
$
	f(\B) = f(\B \cup \B^*) = f(\S),
$
where the first equality is by the locality property of $f$ and induction, and the second equality holds since $\B^*$ is a basis for $\S$. However, this is in contradiction with the fact that $\V \neq \emptyset$. 

Let us now define $\B_{i}$ as the basis of the $\eps$-net computed in the $i$-th successful iteration.  For any $j \in [k]$, 
let $a_j$ be the number of iterations $i$ such that $B_j \in \B^*$ violates $\B_i$. That is, 
\begin{equation*}
a_j = \abs{\{i \in [t] \ |\ f(\B_{i}) < f(\B_{i} \cup \{B_j\})\}}.
\end{equation*}
Since $\V \cap \B^* \neq \emptyset$ in each of the first $t$ successful iterations, there must exist at least one $B_j$ which violates $\B_{i}$ for each $j \in [t]$. We thus have 
$
\sum_{j=1}^k a_j \ge t.
$
Moreover, by the weight update rule of the algorithm, we can write the weight of $\B^*$ as
$w_{t}(\B^*) = \sum_{j=1}^k \left(n^{1/r}\right)^{a_j}.$ By combining these and Jensen's inequality we have
\begin{equation*}
w_{t}(\B^*) \ge k\paren{n^{1/r}}^{\sum_{j=1}^{k}a_j/k}  \geq \paren{n^{1/r}}^{t/k}  \ge n^{t/\nu r},
\end{equation*}
since $k \leq \nu$. This concludes the proof of Claim~\ref{cla:upper}.
\end{proof}

\begin{claim}
\label{cla:lower}
For any integer $t \ge 1$, we have $w_{t}(\S) \le e^{t/10\nu} \cdot n$.
\end{claim}
\begin{proof}
For any iteration $t \geq 1$, the weight update procedure at Line~\ref{ln:a-5} of Algorithm~\ref{alg:meta} gives 
\begin{equation}
\label{eq:a-5}
w_{t+1}(\S) = w_t(\S) + (n^{1/r}-1) \cdot w_t(\V) \leq w_t(\S) + (n^{1/r}) \cdot w_t(\V).
\end{equation}
Moreover, by the condition at Line~\ref{ln:a-4} of the algorithm, we have,
\begin{equation}
\label{eq:a-6}
w_{t}(\V) \le \eps \cdot w_t(\S) = \frac{1}{10 \nu \cdot n ^{1/r}} \cdot w_t(\S),
\end{equation}
by the choice of $\eps$ in the algorithm. Combining (\ref{eq:a-5}) and (\ref{eq:a-6}) we have
\begin{align*}
w_{t}(\S) \le \left(1 + \frac{1}{10\nu}\right)^{t} \cdot w_0(\S) \le e^{t/10\nu} \cdot n. \qed
\end{align*}

\end{proof}
\noindent
We get back to the analysis of the number of iterations. By Eq~(\ref{eq:a-1}) we have $n^{t/\nu r} \le e^{t/10\nu} n$, hence,
$\frac{t}{\nu} \le \frac{10 r \ln n}{10 \ln n - r}.$
Since $r \le \ln n$, we have $\frac{t}{\nu} \le \frac{10}{9} r$.  Therefore the number of successful iterations cannot exceed $\frac{10}{9} \nu r$, and hence the total number of iterations is bounded by $\frac{20}{9} \nu r$ with probability $1 - e^{-\Omega(\nu r)}$.
\end{proof}

\begin{remark}\label{rem:monte-carlo}
We can easily turn our Las-Vegas algorithm in this section (Algorithm~\ref{alg:meta}) into a Monte-Carlo algorithm by the following modifications: First we pick an $\eps$-net of size $m_{\eps, \lambda_\S, 1/(n \nu)}$, and second, the algorithm return ``FAIL'' whenever $w(\V) > \eps w(\S)$, which will not happen in the first $O(\nu r) = O(\nu \log n)$ iterations with probability at least $1 - \nu \log n \cdot 1/(n \nu) \ge 1 - o(1)$.
\end{remark}

\subsection{Implementation in the Streaming Model}
\label{sec:streaming}

Starting from this section, we show how to implement Algorithm~\ref{alg:meta} in the three big data models considered in the paper. We start with the streaming algorithm.   
In the multi-pass streaming model the elements of $\S$ arrive one by one, and $f(\cdot)$ is known to the algorithm at the beginning. We allow the algorithm to make multiple linear scans of the input.

The main challenge in the streaming implementation of Algorithm~\ref{alg:meta} is that we cannot afford to store the weights of all elements in $\S$ which are needed in the $\eps$-net sampling.  To resolve this issue, we instead store the set of bases computed at 
all the \emph{successful} iterations -- these are the only iterations that we change the weight function -- in a collection $\mathscr{B}$, using which we can compute the weight of each element of $\S$ {\em on the fly}. In particular, the weight of a set $S_i \in \S$ in 
iteration $j$ of the algorithm, namely, $w_j(S_i)$, is computed as $w_j(S_i) := (n^{1/r})^{a_i}$ where $a_i := \card{\set{\B \in \mathscr{B} \mid f(\B \cup \set{S_i}) > f(\B)}}$. It is immediate to verify that this indeed implements the same weight
function in Algorithm~\ref{alg:meta}. It is also easy to see that having access to these weights, we can sample each set with probability proportional to its weight using the weighted version of reservoir sampling~\cite{Chao82},
and hence implement each iteration of Algorithm~\ref{alg:meta} in one pass over the stream. 

The rest of Algorithm~\ref{alg:meta} can be implemented in the streaming model in a straightforward way. 
Let $T_b(m)$ be the time complexity of computing a basis for a set of size $m$, and $T_v(t, b)$ be the time complexity of finding all elements in a set $\T \subseteq \S$ of size $t$ which violate a set $\B$ of 
size $b$, i.e., all $S \in \S$ such that $f(\B \cup S) > f(\B)$. This allows us to prove the following theorem.  

\begin{theorem}
\label{thm:streaming}
Suppose $(\S,f)$ is an LP-type problem with combinatorial dimension $\nu$, VC dimension $\lambda$, and bit-complexity $\bit{\S}$ for each element of $\S$. 
For any integer $r \leq \ln{n}$, we can compute $f(\S)$ with high probability in the streaming model, using $O(\nu r)$ passes, and
$\Ot(\lambda n^{1/r} \cdot \nu + \nu^2) \cdot \bit{\S}$ space. The total running time of the algorithm is also $O(\nu r \cdot T_v(n, \nu) + \nu  r \cdot T_b(\lambda n^{1/r} \cdot \nu))$. 
\end{theorem}
\begin{proof}
The correctness of the algorithm follows from Lemma~\ref{lem:main-correct}. As each iteration of Algorithm~\ref{alg:meta} can be implemented in one pass, the total number of passes needed by our streaming 
algorithm is $O(\nu r)$ with high probability by Lemma~\ref{lem:main-iteration}. 

Recall that the size of each $\eps$-net $\N$ sampled in Algorithm~\ref{alg:meta} is
$
m = m_{\eps, \lambda, \frac{2}{3}} = \Ot\paren{\lambda \nu n^{1/r}},
$
by the choice of $\eps$ in the algorithm and $m_{\eps,\lambda,\frac{2}{3}}$ in Lemma~\ref{lem:eps-net}.  
The space needed by the algorithm to store $\N$ in each iteration is $O(m) \cdot \bit{\S}$, which is equal to $\Ot\paren{\lambda \nu n^{1/r}} \cdot \bit{\S}$ bits. 
We also need to store all bases in successful iterations, which requires $O(\nu \cdot r) \cdot O(\nu) \cdot \bit{\S} = \Ot(\nu^2) \cdot \bit{\S}$ (since $r  = O(\log{n})$) 
as each basis requires $O(\nu) \cdot \bit{\S}$ bits to represent and there are total of $O(\nu r)$ such bases. 

Each pass of the algorithm involves performing a violation test over the $n$ elements of $\S$, which takes $O(T_v(n,\nu))$ time. And computing a basis of $m$ elements which takes $O(T_b(m))$ times. 
The run-time follows by multiplying these numbers by the number of passes, and by  choice of $m$.
\end{proof}

\subsection{Implementation in the Coordinator Model}
\label{sec:coordinator}

Recall that in the coordinator model the input set $\S$ is arbitrarily partitioned among $k$ sites $P_1, \ldots, P_k$ such that for any $i \in [k]$, the site $P_i$ receives the elements $\S_i$. The $k$ sites and the coordinator want to jointly
compute $f(\S) = f(\S_1 \cup \cdots \cup \S_k)$ via communication. The function $f$ is a public knowledge, that is, all parties know how to evaluate the function $f(\T)$ for any $\T \in 2^{\S}$ assuming $\T$ resides entirely on that machine.


Similar to the streaming model, the main step here is also the implementation of the $\eps$-net sampling procedure in  Algorithm~\ref{alg:meta}.

\begin{lemma}
	\label{lem:net-coordinator}
	
	The coordinator can sample a subset $\N \subseteq \S$ of size $m$ according to the weight function $w : \S \to \mathbb{R}$ using $2$ rounds and $O(m \cdot \bit{\S} + k(\ell/r+1) \log n)$ bits of communication, where $\ell$ is the number of times the weight function $w(\cdot)$ has been updated when simulating Algorithm~\ref{alg:meta} in the coordinator model.
\end{lemma}

\begin{proof}
	The sampling algorithm is as follows.  In the first round each site $P_i$ sends $w(\S_i)$ to the coordinator.  Note that $w(\S_i)$ for any $i \in [k]$ can be described in $\log(1 + n^{1/r})^\ell = O(\ell/r \cdot \log n)$ bits.  
	
	In the second round the coordinator generates $m$ i.i.d.\ random numbers $x_1, \ldots, x_m$ from $[k]$ from the distribution
	$\Pr[i \text{ is sampled}] = \frac{w(\S_i)}{w(\S)}$,
	and sends the $i$-th site the number $y_i = \abs{\{j\ |\ x_j = i\}}$.
	After obtaining $y_i$, site $P_i$ samples $y_i$ elements from its local set $\S_i$ according to the distribution $\Pr[S \text{ is sampled}] = \frac{w(S)}{w(\S_i)}$, and sends the sampled elements to the coordinator. 
	Note that $y_i \le m \le n$ for any $i \in [m]$, and thus the communication cost of this round is bounded by $O(k) \cdot O(\ell/r \cdot \log n) + O(m) \cdot \bit{\S}$ bits.
	
	Finally, the sampling is indeed with respect to the weight function $w(\cdot)$, since 
	\[
	\Pr[S \text{ is sampled}] =  \frac{w(\S_i)}{w(\S)} \cdot \frac{w(S)}{w(\S_i)} =  \frac{w(S)}{w(\S)}. 
	\]
	This concludes the proof. 
\end{proof}

In order to implement Algorithm~\ref{alg:meta}, each site should also be able to determine the set of violating elements in its input. This can be done easily by asking the coordinator to share the basis computed in each iteration with every site. 
The proof of Theorem~\ref{thm:coordinator} follows directly from that of Theorem~\ref{thm:streaming} by plugging in Lemma~\ref{lem:net-coordinator}.

\begin{theorem}
\label{thm:coordinator}
Suppose $(\S,f)$ is an LP-type problem with combinatorial dimension $\nu$, VC dimension $\lambda$, and bit-complexity $\bit{\S}$ for each element of $\S$. 
For any integer $r \leq \ln{n}$, we can compute $f(\S)$ with high probability in the coordinator model with $k \geq 2$ machines, using $O(\nu r)$ rounds, and
$\Ot(\lambda n^{1/r} \cdot \nu^2 + k \cdot \nu^2) \cdot \bit{\S}$ communication in total. The local computation time of the coordinator is $O(\nu r \cdot (T_b(\lambda n^{1/r} \cdot \nu) + k \nu))$ and the local computation time of the $i$-th site
is $O(\nu r \cdot T_v(n_i, \nu))$ where $n_i := \card{\S_i}$. 
\end{theorem}

\subsection{Implementation in the MPC Model}
\label{sec:MPC}

The implementation of Algorithm~\ref{alg:meta} in the MPC model can be done similarly as that in the coordinator model, by choosing one of the machines to play the role of coordinator. The only problem is that when the number of machines is large, 
the machines cannot simply send all the messages to the coordinator directly, as it will blow up the load in the coordinator. 

Our general strategy is to simulate our implementation of the meta-algorithm for the coordinator model in the MPC model for $r=1/\delta$ round protocols. 
The main challenge in implementing this is that once we require the load of roughly $n^{\delta}$ per machine, we need to start with $k=n^{1-\delta}$ 
machines to begin with to fit the whole input across all machines. This means that the number of sites in the simulation is $k$. But then, if all these machines need to send even one bit to the designated coordinator machine (or vice versa), this requires
a load of $n^{1-\delta}$ on the coordinator machine which is prohibitively large for any $\delta < 1/2$. 

In order to fix this, we are going to use the by now standard approach of~\cite{GSZ11}. There are only two steps that the coordinator and the machines need to communicate with each other: (1) when the machines need to send a 
sample of the $\eps$-net, and (2) when the coordinator needs to send the basis to the machines. The latter can be done easily in $O(1/\delta)$ MPC rounds on machines of memory $O(n^{\delta})$: the coordinator first shares this information with 
$n^{\delta}$ other machines in one round; each of these machines next shares this information with another set of $n^{\delta}$ machines (unique to each original machine). In $O(1/\delta)$ rounds all the $n^{1-\delta}$ machines
would receive this information (see~\cite{GSZ11} for more details on this general approach).  

To handle the part when the machines need to send the $\eps$-net $\N$ to the coordinator, we do as follows. Recall that the size of $\N$ is at most $\tilde{O}(\lambda n^\delta \nu^2)$, and thus it will fit the memory of the coordinator. However, we first need to sample this
according to the correct distribution.  In order to do this, we use our approach for implementing the streaming algorithm. Since by the previous part we managed to share the basis computed in each iteration with every machine, as in the case of streaming algorithms, the machines can compute the weights of every constraint they have. The total weight of the constraints can also be computed in $O(1/\delta)$ rounds using the sort and search method of~\cite{GSZ11}. As a result, each machine can locally perform the sampling of $\N$ and send this information to the coordinator. 
%
%
To summarize, we have the following theorem. 

\begin{theorem}
\label{thm:MPC}
Suppose $(\S,f)$ is an LP-type problem with combinatorial dimension $\nu$, VC dimension $\lambda$, and bit-complexity $\bit{\S}$ for each element of $\S$. 
For any $\delta \in (0,1)$, we can compute $f(\S)$ with high probability in the MPC model using $O(\nu/\delta^2)$ rounds with
$\Ot(\lambda n^{\delta} \cdot \nu^2) \cdot \bit{\S}$ load per machine. 
\end{theorem}

%% file: application.tex
\newcommand{\TLE}{\ensuremath{\textnormal{T}_{\mathtt{LE}}}}

\newcommand{\TLP}{\ensuremath{\textnormal{T}_{\mathtt{LP}}}}

\newcommand{\TMEB}{\ensuremath{\textnormal{T}_{\mathtt{MEB}}}}

\newcommand{\TSVM}{\ensuremath{\textnormal{T}_{\mathtt{SVM}}}}

\section{Examples and Applications}
\label{sec:app}

We now give examples of the application of our algorithms for general LP-type problems.  We will discuss several fundamental optimization problems in machine learning, namely, linear programming, 
Linear SVM, and Core SVM. 
Recall that when implementing our meta algorithm in each model, we have left two functions $T_v(\cdot)$ (the time needed for performing the violation test) and $T_b(\cdot)$ (the time for computing the basis) unspecified.  In this section
we will provide concrete bounds for these functions in the context of the concrete problems we study.
Throughout this section, we assume that the bit-complexity of each number in the input is $O(\log{n})$ bits.

\subsection{Linear Programming}\label{sec:LP}

A linear program is an optimization problem of the type: 
\begin{align}
\min_{x \in \IR^d} \sum_{i=1}^{d} c_i x_i \quad \textnormal{subject to} \quad \sum_{i=1}^d a_i^j x_i \le b^j \quad \textnormal{for all $j \in [n]$}. \label{eq:LP}
\end{align} 

A $d$-dimensional linear program can be modeled as an LP-type problem as follows. Let $\S$ be a set family of size $n$ such that
for every constraint in (\ref{eq:LP}), there exists a unique element $S \in \S$ which is the half-space in the $d$-dimensional Euclidean space $\IR^d$ containing the points that satisfy this single constraint. 
We define the function $f$ over subsets of $\S$ such that for every $\A \subseteq \S$, $f(\A)$ is the lexicographically smallest point that minimizes the objective value of LP while satisfying only the constraints in $\A$.
The linear program (\ref{eq:LP}) now corresponds to the LP-type problem $(\S,f)$ (we use $\S$ as opposed to our previous notation $S$, since each element of $\S$ is now itself a subset of $\IR^d$, and hence $\S$ forms a set family). 
We refer the interested readers to~\cite{MSW96} for more details on connection between linear programming and LP-type problems.  

It is known that the combinatorial dimension $\nu$ of this particular LP-type problem $(\S, f)$ is at most $d+1$~\cite{MSW96}. The VC dimension $\lambda$ is also at most $d+1$~\cite{VC15}.

In the following, let $\TLP(m,t)$ denotes the time needed to solve a linear program with $\Theta(m)$ constraints and 
$\Theta(t)$ variables.

\begin{proposition}\label{prop:LP-oracles}
	For any linear program with $n$ constraints and dimension $d$:
	\begin{itemize}
		\item The time needed to compute a basis of $m \geq d$ given constraints is 
		$
		T_b(m) = O(d \cdot \TLP(m,d)). 
		$
		\item The time needed to compute all constraints  that violate a given basis of size $b = O(d)$ among $t$ constraints is
		$
		T_{v}(t,b) =  O(t \cdot d + d\cdot\TLP(d, d)). 
		$
	\end{itemize} 
\end{proposition}
\begin{proof}
	To find a basis $\B$ of a set $\N$ of $m$ constraints, we first solve the LP only given the constraints in $\N$ to obtain a point $x^*=(x^*_1,\ldots,x^*_d)$ with optimal value $c^*$. Recall that in our mapping of LP to an LP-type problem, we
	need to find a lexicographically smallest optimal solution on constraints in $\N$, which may not be the point $x^*$ even though the objective value is still $c^*$.  Hence, we now write a separate linear program:
	\begin{align*}
	&\min_{x \in \IR^d} x_1  \\ 
	&\textnormal{subject to} \quad \sum_{i=1}^{d} c_i x_i = c^* \quad \textnormal{and} \quad \sum_{i=1}^d a_i^j x_i \le b^j \quad \textnormal{for all $j \in \N$}.
	\end{align*}
	This allows us to find an optimal solution to the LP with the minimum value of $x_1$. Repeating this procedure for $d$ iterations and for $i$-th iteration fixing $x_1,..,x_{i-1}$ computed so far, and finding the minimum value for $x_i$, 
	allows us to find the lexicographically smallest optimal solution. These LPs all are $d$-dimensional with $\Theta(m)$ constraints, and hence can be solved in $O(d \cdot \TLP(m,d))$ time in total, finalizing the first part.

	A basis of size $b$ in a linear program consists of $b$ constraints of the LP that are all tight by the assignment of the variables. Hence, given the basis, we only need to solve the linear program on a system of $b$ linear inequalities to determine a value of $x^*$ that 
	is tight for all the constraints in the basis. This can be done in $O(d\TLP(d, d))$ time (as we do before). After this, we can simply check the $d$-dimensional vector $x^*$ against all the $t$ constraints and add each one as a violating set if $x^*$ does not satisfy the constraint
	in $O(t \cdot d)$ time, finalizing the second part.
\end{proof}

Plugging in the currently best known bound for $\TLP(m,d) = \Ot(\sqrt{d}\paren{dm+d^{2.373}})$ by~\cite{LS14} in Proposition~\ref{prop:LP-oracles}, and the aforementioned bounds on $\nu,\lambda=O(d)$, we can prove 
the following theorem using Theorems~\ref{thm:streaming},~\ref{thm:coordinator}, and~\ref{thm:MPC}. 

\begin{theorem}\label{thm:LP}
	We give the following randomized  algorithms  for $d$-dimensional linear programming with $n$ constraints. For any  $r \geq 1$ and $\delta \in (0,1)$:
	\begin{itemize}[leftmargin=12pt]
		\item \emph{Streaming:} An $O(d \cdot r)$-pass algorithm with $\Ot(d^3 \cdot n^{1/r})$ space in $\Ot(n) \cdot \poly(d)$ time. 
		\item \emph{Coordinator:} An $O(d \cdot r)$-round algorithm with $\Ot(d^4 n^{1/r} + d^3 k)$ total communication in which the coordinator and each site $i \in [k]$ spend $\Ot(n^{1/r}+k) \cdot \poly(d)$ time and $\Ot(n_i) \cdot \poly(d)$ time, 
		respectively, where $n_i$ is the number of constraints on site $i$.  
		\item \emph{MPC:} An $O(d/\delta^2)$-round algorithm with $\Ot(d^3n^{\delta})$ load per machine and $\Ot(n) \cdot \poly(d)$ time in total. 
	\end{itemize}
\end{theorem}

\subsection{Linear Support Vector Machine}\label{sec:LSVM}

In Linear Support Vector Machine (SVM) problem~\cite{BGV92},
we have a set of tuples $\{(x_1, y_1), \ldots, (x_n, y_n)\}$ such that for each index $j \in [n]$, $x_j \in \mathbb{R}^d$ and $y_j \in \set{-1, +1}$. The goal is to compute a hyperplane $u = (u_1, \ldots, u_d)$ which is the outcome of the following quadratic optimization problem~\cite{BGV92}:
\begin{align}
\min_{u \in \IR^d}  \norm{u}^2_2  \quad \textnormal{subject to} \quad y_j \cdot \inner{u}{x_j} \geq 1 \quad \textnormal{for all $j \in [n]$}. \label{eq:LSVM}
\end{align}

From a geometrical point of view, the problem~(\ref{eq:LSVM}) corresponds to finding a hyperplane which separates the set of point $\set{x_1, \dotsc, x_n}$ according to their labels with the maximum margin value (if possible); see, e.g., \cite{BGV92} for more 
information on this fundamental problem.\footnote{Our algorithm works effectively for the hard-margin Linear SVM. In the case of the soft-margin Linear SVM, the optimization problem can also be formulated in the form of LP-type problem, but the dimension of such formulation is large -- proportional to the size of input.}  Note that the problem~(\ref{eq:LSVM}) is {\em not} a linear program. However, one can show that it is an LP-type problem $(\S,f)$ where $\S$ is a set family in $\IR^d$ in which every set
contains the points that satisfy a particular constraint, and $f(\A)$ for $\A \subseteq \S$ computes the optimal solution of (\ref{eq:LSVM}) given only the constraints to $\A$~\cite{MSW96} (unlike linear programming, the optimal solution to (\ref{eq:LSVM}) under any 
set of constraints is unique and hence we do not need the lexicographically first constraint).

The combinatorial dimension of $(\S,f)$ is $\nu \leq d+1$~\cite{MSW96}, and the VC dimension of $(\S,\IR^d)$ is $\lambda \leq d+1$~\cite{VC15}.  
In the following, let $\TSVM(m,d)$ denote the time needed to solve an instance of Linear SVM problem with $m$ constraints and $d$ variables. 
We show how to implement the basis computation and violation test for Linear SVM in the following proposition. 

\begin{proposition}\label{prop:LSVM-oracles}
	For any Linear SVM problem with $n$ constraints and dimension $d$:
	\begin{itemize}
		\item The time needed to compute a basis of $m \geq d$ given constraints is 
		$
		T_b(m) = O(\TSVM(m,d)). 
		$
		\item The time needed to compute all constraints  that violate a given basis of size $b$ among $t$ constraints is
		$
		T_{v}(t,b) =  O(t \cdot d + \TSVM(d, d)). 
		$
	\end{itemize} 
\end{proposition}
\begin{proof}
	To find a basis $\B$ of a set $\N$ of $m$ constraints, we simply need to solve another instance of Linear SVM, i.e., (\ref{eq:LSVM}), only on the given constraints. This can be done in $O(\TSVM(m,d))$ by definition. 
	The second part can also be solved by solving a linear equation exactly as in the case in Proposition~\ref{prop:LP-oracles}. 
\end{proof}

Plugging in the currently best known bound for $\TSVM(m,d) = O((m+d)^3)$ by quadratic programming in~\cite{YT89} in Proposition~\ref{prop:LSVM-oracles}, and the aforementioned
bounds on $\nu,\lambda=O(d)$, we can prove Theorem~\ref{thm:LSVM} using Theorems~\ref{thm:streaming},~\ref{thm:coordinator}, and~\ref{thm:MPC}.

\begin{theorem}\label{thm:LSVM}
	We give the following randomized algorithms for $d$-dimensional linear support vector machine problem with $n$ constraints. For any $r \geq 1$ and $\delta \in (0,1)$:
	\begin{itemize}[leftmargin=12pt]
		\item \emph{Streaming:} An $O(d \cdot r)$-pass algorithm with $\Ot(d^3 \cdot n^{1/r})$ space in $\Ot(n) \cdot \poly(d)$ time. 
		\item \emph{Coordinator:} An $O(d \cdot r)$-round algorithm with $\Ot(d^4 n^{1/r} + d^3 k)$ total communication in which the coordinator and each site $i \in [k]$ spend $\Ot(n^{3/r}+k) \cdot \poly(d)$ time and $\Ot(n_i) \cdot \poly(d)$ time, 
		respectively, where $n_i$ is the number of constraints on site $i$.  
		\item \emph{MPC:} An $O(d/\delta^2)$-round algorithm with $\Ot(d^3n^{\delta})$ load per machine and $\Ot(n + n^{3\delta}) \cdot \poly(d)$ time in total. 
	\end{itemize}
\end{theorem}

\subsection{Core Vector Machine}\label{sec:CSVM}
Tsang at el.~\cite{TKC05} proposed  \emph{core vector machines} as a way of speeding up kernel methods in SVM training (see~\cite{BGV92}). This is achieved by reformulating the original kernel method as
an instance of the \emph{minimum enclosing ball} (MEB) problem, defined as follows: Given a set of points $P := \set{p_1, \ldots, p_n}$ in $\IR^d$, find a center $p$ and a minimum radius $r$ such that all the points in $P$ are within a $d$-dimensional sphere of radius $r$ centered at $p$. MEB can be formulated as the following optimization problem: 
\begin{align}
\min_{r \in \IR, p \in \IR^d}  r  \quad \textnormal{subject to} \quad  \norm{p - p_j}_2 \leq r \quad \textnormal{for all $j \in [n]$}. \label{eq:MEB}
\end{align} 

This problem is also an LP-type problem $(\S,f)$ formulated similarly to linear programming and Linear SVM~\cite{MSW96}.

The combinatorial dimension of $(\S,f)$ is $\nu \leq d + 1$~\cite{MSW96} and the VC dimension of $(\S,\IR^d)$ is $\lambda \leq d+1$~\cite{WD81}. 
Let $\TMEB(m,d)$ denote the time needed to solve an instance of MEB problem with $m$ constraints and $d$ variables. 
The following proposition show how to implement the basis computation and violation test for MEB (the proof is identical to Proposition~\ref{prop:LSVM-oracles} and is hence omitted). 

\begin{proposition}\label{prop:CSVM-oracles}
	For any Linear SVM problem with $n$ constraints and dimension $d$:
	\begin{itemize}
		\item The time needed to compute a basis of $m \geq d$ given constraints is 
		$
		T_b(m) = O(\TMEB(m,d)). 
		$
		\item The time needed to compute all constraints  that violate a given basis of size $b$ among $t$ constraints is
		$
		T_{v}(t,b) =  O(t \cdot d + \TMEB(d, d)). 
		$
	\end{itemize} 
\end{proposition}

As MEB can be cast as a convex quadratic program, we have $\TMEB(m,d) = O((m + d)^3)$ by~\cite{YT89} as before. Hence, Theorems~\ref{thm:streaming},~\ref{thm:coordinator}, and~\ref{thm:MPC} imply the following 
result. 

\begin{theorem}\label{thm:coreSVM}
	We give the following randomized  algorithms for $d$-dimensional core vector machine problem with $n$ constraints. For any integer $r \geq 1$:
	\begin{itemize}[leftmargin=12pt]
		\item \emph{Streaming:} An $O(d \cdot r)$-pass algorithm with $\Ot(d^3 \cdot n^{1/r})$ space in $\Ot(n + n^{3/r}) \cdot \poly(d)$ time. 
		\item \emph{Coordinator:} An $O(d \cdot r)$-round algorithm with $\Ot(d^4 n^{1/r} + d^3 k)$ total communication in which the coordinator and each site $i \in [k]$ spend $\Ot(n^{3/r}+k) \cdot \poly(d)$ time and $\Ot(n_i) \cdot \poly(d)$ time, 
		respectively, where $n_i$ is the number of constraints on site $i$.  
		\item \emph{MPC:} An $O(d/\delta^2)$-round algorithm with $\Ot(d^3n^{\delta})$ load per machine and $\Ot(n+n^{3\delta}) \cdot \poly(d)$ time in total. 
	\end{itemize}
\end{theorem}

%% file: lb.tex
\newcommand{\TCI}{\ensuremath{\textnormal{\textsf{TCI}}}}

\section{Lower Bounds}\label{sec:lower}

In this section we prove information-theoretic lower bounds for linear programming that hold against \emph{any} algorithm. We obtain our lower bounds by establishing the \emph{communication complexity} for $2$-dimensional linear programming, and then translating it to lower bounds in the big data models. 
In the following, we first give some background on communication complexity and then present an intermediate 
problem, called two-curve intersection problem (\TCI), that we consider en route to proving our result for linear programming. We then prove a lower bound for $\TCI$ and present its implications for linear programming
in the streaming and coordinator models.


\input{cc.tex}

\input{problem.tex}
\input{CC-TCI.tex}

%% file: cc.tex
\subsection{Background}\label{sec:cc-background}

\paragraph{Communication Complexity.} 
We focus on the standard two-party communication complexity model of Yao~\cite{Yao79}. In this model, Alice and Bob receive an input $X \in \mathcal{X}$ and $Y \in \mathcal{Y}$, respectively. 
In an $r$-round protocol, Alice and Bob can communicate 
up to $r$ messages with each other. In particular, for an even $r$, Bob first sends a message to Alice, followed by a message from Alice to Bob, and so on, until Bob receives the last message and outputs the answer. 
For an odd $r$, the only difference is that Alice starts first and then the players continue like before until Bob outputs the answer. 

The communication complexity of a problem $P : \mathcal{X} \times \mathcal{Y} \rightarrow \mathcal{Z}$, denoted by $\CC{P}$, is the minimum worst-case communication cost of any protocol (possibly randomized) that can solve $P$ with probability at least $2/3$. The $r$-round communication complexity of $P$, denoted
by $\CCr{P}{r}$, is similarly defined with respect to protocols that are allowed at most $r$ rounds of communication. 

\medskip

\emph{Augmented Indexing.} In the Augmented Indexing Problem, denoted by $\AugIndex_n$, Alice is given a binary string $x \in \set{0,1}^{n}$, and Bob is given an index $i \in \set{0,1}$ plus the first $i-1$ bits of the string $x$, i.e., $x_1,\ldots,x_{i-1}$. 
The goal is for Bob to output the bit $x_i$. It is well-known that $1$-round communication complexity of this problem is $\CCr{\AugIndex_n}{1} = \Omega(n)$ (see, e.g.~\cite{MiltersenNSW98}). 

\paragraph{Information Theory.}
Throughout this section,  we use bold-face fonts, say $\rA$, to denote random variables, and normal font, say $A$, to denote their realizations. 
For a random variable $\rA$, $\supp{\rA}$ denotes its support and 
$\distribution{\rA}$ its distribution. We sometimes abuse the notation and use $\rA$ and  $\distribution{\rA}$ interchangeably. Furthermore, for a $t$-tuple $(X_1,\ldots,X_t)$ and any integer $i \in [t]$, 
we define $X^{<i} := (X_1,\ldots,X_{i-1})$ and $X^{>i} := (X_{i+1},\ldots,X_{t})$. 

Our proof relies on basic concepts from information theory, which we review briefly here. For a broader introduction, we refer the interested reader to the excellent text by Cover and Thomas~\cite{ITbook}. 

\paragraph{Entropy and Mutual Information.} The Shannon entropy of $\rA$ is defined as 
\[
\en{\rA} := \sum_{A \in \supp{\rA}} \Pr\paren{\rA=A} \cdot \log{\paren{1/\Pr\paren{\rA=A}}}.
\]
The conditional entropy of $\rA$ on random variable $\rB$ is defined as $\en{\rA \mid \rB} := \Ex_{B \sim \rB}\bracket{\en{\rA \mid \rB=B}}$. The (conditional) mutual information between $\rA$ and $\rB$
is $\mi{\rA}{\rB \mid \rC} := \en{\rA \mid \rC} - \en{\rA \mid \rB,\rC}$. We shall use the following basic properties of entropy and mutual information throughout. 

\begin{fact}[cf.~\cite{ITbook}; Chapter~2]\label{fact:it-facts}
	Let $\rA$, $\rB$, $\rC$, and $\rD$ be four (possibly correlated) random variables.
	\begin{enumerate}
		\item \label{part:uniform} $0 \leq \en{\rA} \leq \log{\card{\supp{\rA}}}$. The right equality holds
		iff $\distribution{\rA}$ is uniform.
		\item \label{part:info-zero} $\mi{\rA}{\rB} \geq 0$. The equality holds iff $\rA$ and
		$\rB$ are \emph{independent}.
		\item \label{part:cond-reduce} \emph{Conditioning on a random variable can only reduce the entropy}:
		$\en{\rA \mid \rB,\rC} \leq \en{\rA \mid  \rB}$.  
		The equality holds iff $\rA \perp \rC \mid \rB$.
		\item \label{part:chain-rule} \emph{Chain rule for mutual information}: $\mi{\rA,\rB}{\rC \mid \rD} = \mi{\rA}{\rC \mid \rD} + \mi{\rB}{\rC \mid  \rA,\rD}$.
	\end{enumerate}
\end{fact}

\paragraph{Measures of Distance Between Distributions.} 
For two distributions $\mu$ and $\nu$, the \emph{Kullback-Leibler divergence} between $\mu$ and $\nu$ is denoted by $\kl{\mu}{\nu}$ and defined as: 
\begin{align}
\kl{\mu}{\nu}:= \Ex_{a \sim \mu}\Bracket{\log\frac{\Pr_\mu(a)}{\Pr_{\nu}(a)}}. \label{eq:kl}
\end{align}
We have the following relation between mutual information and KL-divergence. 
\begin{fact}\label{fact:kl-info}
	For random variables $\rA,\rB,\rC$, 
	\[\mi{\rA}{\rB \mid \rC} = \Ex_{(b,c) \sim {(\rB,\rC)}}\Bracket{ \kl{\distribution{\rA \mid \rC=c}}{\distribution{\rA \mid \rB=b,\rC=c}}}.\] 
\end{fact}

We denote the total variation distance between two distributions $\mu$ and $\nu$ on the same 
support $\Omega$ by $\tvd{\mu}{\nu}$, defined as: 
\begin{align}
\tvd{\mu}{\nu}:= \max_{\Omega' \subseteq \Omega} \paren{\mu(\Omega')-\nu(\Omega')} = \frac{1}{2} \cdot \sum_{x \in \Omega} \card{\mu(x) - \nu(x)}.  \label{eq:tvd}
\end{align}
\noindent
We use the following basic properties of total variation distance. 
\begin{fact}\label{fact:tvd-small}
	Suppose $\mu$ and $\nu$ are two distributions for $\event$, then, 
	$
	\Pr_{\mu}(\event) \leq \Pr_{\nu}(\event) + \tvd{\mu}{\nu}.
	$
\end{fact}

The following Pinskers' inequality bounds the total variation distance between two distributions based on their KL-divergence, 

\begin{fact}[Pinsker's inequality]\label{fact:pinskers}
	For any distributions $\mu$ and $\nu$, 
	$
	\tvd{\mu}{\nu} \leq \sqrt{\frac{1}{2} \cdot \kl{\mu}{\nu}}.
	$ 
\end{fact}

%% file: problem.tex
\renewcommand{\AA}{\ensuremath{{A}}}
\newcommand{\BB}{\ensuremath{{B}}}

\newcommand{\seq}[1]{\ensuremath{\langle #1 \rangle}}

\subsection{The Two-Curve Intersection Problem (\TCI)}\label{sec:tci-problem}
We consider the following problem, whose lower bound implies a lower bound for linear programming in the two-dimensional Euclidean space (as we show shortly). 

Alice and Bob are given sequences of $n$ numbers $\AA := \seq{a_1,\ldots,a_n}$ and $\BB := \seq{b_1,\ldots,b_n}$ in $\IQ^n$, respectively, such that: 
\begin{enumerate}
	\item \emph{Monotonicity:} $\AA$ is monotonically increasing and $\BB$ is monotonically decreasing. 
	\item \emph{Convexity:} For any $i \in [n]$, in $\AA$ we have $a_i - a_{i-1} \leq a_{i+1} - a_i$ and conversely in $\BB$ we have $b_i - b_{i-1} \geq b_{i+1} - b_{i}$. 
\end{enumerate}

\noindent
The goal is to find the \emph{smallest index} $i^* \in [n]$ such that $a_{i^*} \leq b_{i^*}$ but $a_{i^*+1} > b_{i^*+1}$, under the promise that such an index always exists. We can interpret the sequence $\AA$ as a {two-dimensional curve} in $\IR^2$
that goes through the points $(1,a_1),(2,a_2),\cdots(n,a_n)$ (similarly for $\BB$).  
We refer to this problem as the \emph{two-curve intersection} problem and denote it by $\TCI_n$ for sequences of length $n$ (or $\TCI$ in general). See Figure~\ref{fig:TCI-LP-a} for an illustration of this problem. 

\input{Figs/TCI-LP}

\paragraph{Connection to $\bm{2}$-Dimensional Linear Programming.} We can reduce the two-curve intersection problem to an instance of $2$-dimensional linear programming as follows (see Figure~\ref{fig:TCI-LP-b}). Extend each segment of 
the curve in Alice's and Bob's input to obtain a line that defines a constraint in which all points above this line are feasible (blue region for Alice and green region for Bob in Figure~\ref{fig:TCI-LP-b}). 
The feasible region of this linear program is the set of points in $\IR^2$ that lie above both of Alice's and Bob's curve. By minimizing the $y$-axis on the feasible region, we obtain the first ``fractional'' point in which Alice's curve goes above Bob's curve, and by rounding down the $x$-axis of this point, we obtain the index $\istar$ of  $\TCI$. 

\paragraph{Geometric Notations.} We work in the two-dimensional Euclidean space $\IR^2$. We use $p \in \IR^2$ to denote a point, and $p.x$ and $p.y$ to denote its $x$ and $y$ coordinates respectively. Throughout, all the points used
have rational coordinates (i.e., in $\IQ^2$).
For two points $p_1,p_2$ and integers $a \leq b$, we define $\Line(p_1,p_2,a,b)$ as the sequence of $b-a+1$ numbers $\seq{z_a,z_{a+1},\ldots,z_{b}}$ such that for all $i \in [a:b]$, $(i,z_i)$ belongs to the unique line in $\IR^2$ that passes through the points 
$p_1$ and $p_2$. We use the following elementary geometric facts.

\begin{fact}\label{fact:lines}
	Let $\seq{z_a,z_{a+1},\ldots,z_{b}} := \Line(p_1,p_2,a,b)$. 
	\begin{enumerate}
		\item For every $i \in (a:b]$, $z_i - z_{i-1} := \frac{p_2.y - p_1.y}{p_2.x - p_1.x}$.
		\item For every $i \in [a:b]$, $z_i  = \frac{p_2.y - p_1.y}{p_2.x - p_1.x} \cdot (i-p_1.x) + p_1.y =  \frac{p_2.y - p_1.y}{p_2.x - p_1.x} \cdot (i-p_2.x) + p_2.y$. 
	\end{enumerate}
\end{fact}

We also define a notion called \emph{step curve}. For a string $X = (x_1,\ldots,x_m) \in \set{0,1}^m$ and a parameter $\alpha \geq 1$, 
$\Stepcurve(X,\alpha)$ is the sequence of $m+1$ numbers $\seq{z_0,z_1,\ldots,z_m}$ such that $z_0 = 0$ and for all $i \in [m]$, $z_i := z_{i-1} + \alpha+ i + x_{i}$. 

%% file: Figs/TCI-LP.tex
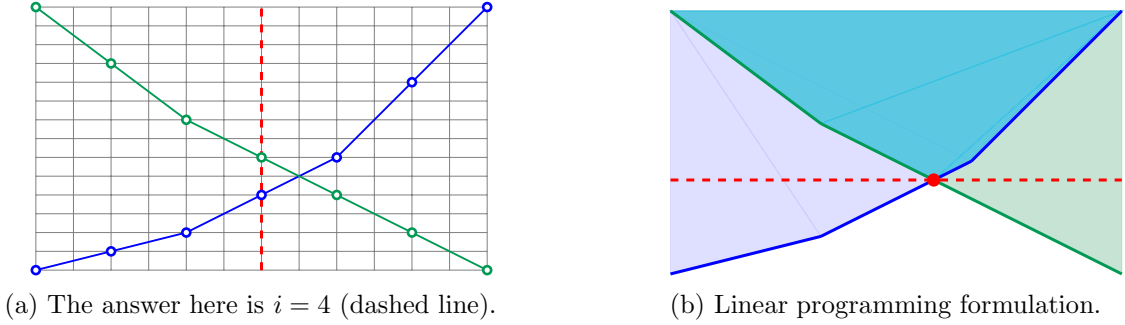
\begin{figure}[h]
    \centering
    \subcaptionbox{The answer here is $i = 4$ (dashed line). \label{fig:TCI-LP-a}}[0.45\textwidth]{

\begin{tikzpicture}[ auto ,node distance =1cm and 2cm , on grid , semithick , state/.style ={ circle ,top color =white , bottom color = white , draw, black , text=black}, every node/.style={inner sep=0,outer sep=0}]

\draw[help lines,xstep=0.5,ystep=.25] (0,0) grid (6,3.5);

\draw[dashed,red, line width=0.4mm] (3,0) -- (3,3.5);

\node[state, circle, blue, line width=0.35mm, minimum height=3pt, minimum width=3pt] (a1) at (0,0){};
\node[state, circle, blue, line width=0.35mm, minimum height=3pt, minimum width=3pt] (a2) at (1,0.25){};
\node[state, circle, blue, line width=0.35mm, minimum height=3pt, minimum width=3pt] (a3) at (2,0.5){};
\node[state, circle, blue, line width=0.35mm, minimum height=3pt, minimum width=3pt] (a4) at (3,1){};
\node[state, circle, blue, line width=0.35mm, minimum height=3pt, minimum width=3pt] (a5) at (4,1.5){};
\node[state, circle, blue, line width=0.35mm, minimum height=3pt, minimum width=3pt] (a6) at (5,2.5){};
\node[state, circle, blue, line width=0.35mm, minimum height=3pt, minimum width=3pt] (a7) at (6,3.5){};

\draw[blue, line width=0.25mm] 
(a1) -- (a2)
(a2) -- (a3)
(a3) -- (a4)
(a4) -- (a5)
(a5) -- (a6)
(a6) -- (a7);

\node[state, circle, ForestGreen, line width=0.35mm, minimum height=3pt, minimum width=3pt] (b1) at (6,0){};
\node[state, circle, ForestGreen, line width=0.35mm, minimum height=3pt, minimum width=3pt] (b2) at (5,0.5){};
\node[state, circle, ForestGreen, line width=0.35mm, minimum height=3pt, minimum width=3pt] (b3) at (4,1){};
\node[state, circle, ForestGreen, line width=0.35mm, minimum height=3pt, minimum width=3pt] (b4) at (3,1.5){};
\node[state, circle, ForestGreen, line width=0.35mm, minimum height=3pt, minimum width=3pt] (b5) at (2,2){};
\node[state, circle, ForestGreen, line width=0.35mm, minimum height=3pt, minimum width=3pt] (b6) at (1,2.75){};
\node[state, circle, ForestGreen, line width=0.35mm, minimum height=3pt, minimum width=3pt] (b7) at (0,3.5){};

\draw[ForestGreen, line width=0.25mm] 
(b1) -- (b2)
(b2) -- (b3)
(b3) -- (b4)
(b4) -- (b5)
(b5) -- (b6)
(b6) -- (b7);

\end{tikzpicture}

}
\hspace{0.75cm}   
    \subcaptionbox{Linear programming formulation. \label{fig:TCI-LP-b}}[0.45\textwidth]{

\begin{tikzpicture}[ auto ,node distance =1cm and 2cm , on grid , semithick , state/.style ={ circle ,top color =white , bottom color = white , draw, black , text=black}, every node/.style={inner sep=0,outer sep=0}]

\node[state, circle, blue, line width=0.35mm, minimum height=0pt, minimum width=0pt] (a1) at (0,0){};
\node[state, circle, blue, line width=0.35mm, minimum height=0pt, minimum width=0pt] (a2) at (1,0.25){};
\node[state, circle, blue, line width=0.35mm, minimum height=0pt, minimum width=0pt] (a3) at (2,0.5){};
\node[state, circle, blue, line width=0.35mm, minimum height=0pt, minimum width=0pt] (a4) at (3,1){};
\node[state, circle, blue, line width=0.35mm, minimum height=0pt, minimum width=0pt] (a5) at (4,1.5){};
\node[state, circle, blue, line width=0.35mm, minimum height=0pt, minimum width=0pt] (a6) at (5,2.5){};
\node[state, circle, blue, line width=0.35mm, minimum height=0pt, minimum width=0pt] (a7) at (6,3.5){};

\node[state, circle, blue, line width=0.35mm, minimum height=0pt, minimum width=0pt] (P1) at (6,1.5){};
\node[state, circle, blue, line width=0.35mm, minimum height=0pt, minimum width=0pt] (P2) at (1,0){};
\node[state, circle, blue, line width=0.35mm, minimum height=0pt, minimum width=0pt] (P3) at (6,2.5){};
\node[state, circle, blue, line width=0.35mm, minimum height=0pt, minimum width=0pt] (P4) at (2.5,0){};


\fill[blue!25, opacity=0.5, line width=0pt, draw] 
(0,0)-- (2,0.5) -- (0,3.5) -- cycle
(2,0.5) -- (4,1.5) -- (0,3.5) -- cycle
(4,1.5) -- (0,3.5) -- (6,3.5) -- cycle;

\node[state, circle, ForestGreen, line width=0.35mm, minimum height=0pt, minimum width=0pt] (b1) at (6,0){};
\node[state, circle, ForestGreen, line width=0.35mm, minimum height=0pt, minimum width=0pt] (b2) at (5,0.5){};
\node[state, circle, ForestGreen, line width=0.35mm, minimum height=0pt, minimum width=0pt] (b3) at (4,1){};
\node[state, circle, ForestGreen, line width=0.35mm, minimum height=0pt, minimum width=0pt] (b4) at (3,1.5){};
\node[state, circle, ForestGreen, line width=0.35mm, minimum height=0pt, minimum width=0pt] (b5) at (2,2){};
\node[state, circle, ForestGreen, line width=0.35mm, minimum height=0pt, minimum width=0pt] (b6) at (1,2.75){};
\node[state, circle, ForestGreen, line width=0.35mm, minimum height=0pt, minimum width=0pt] (b7) at (0,3.5){};

\node[state, circle, blue, line width=0.35mm, minimum height=0pt, minimum width=0pt] (Q1) at (0,3){};
\node[state, circle, blue, line width=0.35mm, minimum height=0pt, minimum width=0pt] (Q2) at (4.65,0){};


\fill[ForestGreen!45, opacity=0.5, line width=0pt, draw] 
(6,3.5)-- (6,0) -- (2,2) -- cycle
(6,3.5) -- (0,3.5) -- (2,2) -- cycle;

\fill[Cyan!75, opacity=0.5, line width=0pt, draw] 
(6,3.5) -- (0,3.5) -- (2,2) -- cycle
(6,3.5) -- (3.5,1.25) -- (2,2) -- cycle
(6,3.5) -- (3.5,1.25) -- (4,1.5) -- cycle;

\draw[blue, line width=0.4mm] 
(a1) -- (a2)
(a2) -- (a3)
(a3) -- (a4)
(a4) -- (a5)
(a5) -- (a6)
(a6) -- (a7);

\draw[ForestGreen, line width=0.4mm] 
(b1) -- (b2)
(b2) -- (b3)
(b3) -- (b4)
(b4) -- (b5)
(b5) -- (b6)
(b6) -- (b7);

\draw[dashed,red, line width=0.4mm] (0,1.25) -- (6,1.25);
\node[circle, red, fill=red, draw, line width=0.35mm, minimum height=4pt, minimum width=4pt] (X) at (3.5,1.25){};

\end{tikzpicture}

}
    
\caption{An illustration of the two-curve intersection problem for $n=7$ points and its connection to $2$-dimensional linear programming.  \label{fig:TCI}}
\end{figure}

%% file: CC-TCI.tex
\subsection{Communication Complexity of \TCI}\label{sec:cc-tci}

Our goal is to prove the following theorem. 
\begin{theorem}\label{thm:cc-tci}
	For any $r \geq 1$, $\CCr{\TCI_n}{r} = \Omega(\frac{1}{r^2} \cdot n^{1/r})$. 
\end{theorem}
\noindent
The proof of Theorem~\ref{thm:cc-tci} is based on an inductive argument, following the general \emph{round-elimination} approach in communication complexity 
(see, e.g.~\cite{MiltersenNSW98,SenV08}). In this approach, one proves the lower bound for $r$-round problems
by showing that a ``too good'' $r$-round protocol will imply a too good $(r-1)$-round protocol, by reducing the $r$-round problem to \emph{multiple} instances of the $(r-1)$-round problem. Following this argument inductively, we will end up with a protocol using only $1$ round. We then directly prove that such a too good $1$-round protocol {\em cannot} exist.

\subsubsection{Base Case: One-Round Protocols}\label{sec:cc-tci-1-round} %

As a warm-up, we first prove Theorem~\ref{thm:cc-tci} for $r=1$, i.e., 1-round protocols.

\begin{lemma}
	\label{lem:cc-tci-1-round}
	$\CCr{\TCI_n}{1} = \Omega(n)$. 
\end{lemma}
\begin{proof}
	We prove this lemma using a reduction from the Augmented Indexing Problem on a universe of size $n-1$. Given an instance of $\AugIndex_{n-1}$ with input $x \in \set{0,1}^{n-1}$ to Alice and $\istar \in [n-1]$ plus $x_1,\ldots,x_{\istar-1}$ to Bob, 
	the players construct the following instance of $\TCI_n$ (without any communication): 
	\begin{enumerate}
		\item Alice creates $\AA :=\seq{a_1,\ldots,a_n} := \Stepcurve(x,0)$. 
		\item Bob creates $\BB := \seq{b_1,\ldots,b_n} = \Line(p_1,p_2,1,n)$, where $p_1 := (n,1)$ and $p_2 := (\istar,a_{\istar}+\istar+1)$. 
		
	\end{enumerate}
	
	The players then run the protocol for $\TCI_n$ on this instance and Bob outputs $x_{\istar} = 1$ (in answer to the $\AugIndex$ instance) iff $\istar$ is returned by the protocol for $\TCI_n$ as the answer on this instance.  
	
	\medskip
	\emph{Correctness of the reduction.} We first verify that the sequences $\AA$ and $\BB$ constructed by Alice and Bob satisfy the promise of the $\TCI_n$ input. By Fact~\ref{fact:lines}, $\BB$ is both monotone
	and convex. It is also easy to see that $\AA$ is monotonically increasing: for all $i \geq 2$, $a_{i} \geq a_{i-1} + i \geq a_{i-1}$. Finally, to verify the convexity of $\AA$, notice that $a_{i} - a_{i-1} = i + x_{i-1} \leq i+1$ while $a_{i+1} - a_{i} = i+1 + x_i \geq i+1$. 
	
	We now prove the correctness of the output in the reduction. Suppose first that $x_{\istar}=0$. In this case, $a_{\istar+1} = a_{\istar} + \istar+1 + x_{\istar} = a_{\istar}+\istar+1 = b_{\istar}$ by definition. 
	On the other hand, $a_{\istar+2} > b_i$ for all $i > \istar$ as $b_i < b_{\istar}$ and $a_{\istar+2} > a_{\istar+1} = b_{\istar}$. 
	As a result, the correct index in $\TCI_n$ is $\istar+1$. Now suppose $x_{\istar}=1$. In this case, $a_{\istar+1} > b_{\istar}$ while 
	$a_{\istar} < b_{\istar}$. As such, the correct index in $\TCI_n$ is $\istar$, finalizing the proof of the correctness of the reduction. 
	
	\medskip
	\emph{Communication cost of the reduction.} The instance of $\TCI_n$ can be created with no communication. As such,  
	\begin{align*}
	\CCr{\TCI_n}{1} \geq \CCr{\AugIndex_{n-1}}{1} = \Omega(n). \qed
	\end{align*}
	
\end{proof}

\subsubsection{General Lower Bound: The Outline}\label{sec:cc-tci-outline} %

We now switch to the main part of the argument in which we prove Theorem~\ref{thm:cc-tci} for all integers $r \geq 1$. In this section we outline our high level approach.
In this section, we will oversimplify many details, and the discussions will be informal for the sake of intuition.

We design a family of distribution $\dist_1,\dist_2,\ldots$, where $\dist_r$ is hard distribution for $r$-round protocols. Distribution $\dist_1$ is the distribution of hard instances obtained in Lemma~\ref{lem:cc-tci-1-round} (from the hard distribution of $\AugIndex$).
Each instance $I$ in the distribution $\dist_r$ is then constructed roughly as follows: we sample $n^{1/r}$ instances from the distribution $\dist_{r-1}$ each over $n^{(r-1)/r}$ points. 
Let us call these instances $I_1,\ldots,I_{n^{1/r}}$. We \emph{embed} these instances inside $I$ so that the following two properties are satisfied: $(i)$ the answer to $\TCI_n$ on instance $I$ is the same as the answer to $\TCI_{n^{(r-1)/r}}$ on instance 
$I_{\zstar}$ for some $\zstar \in [n^{1/r}]$ chosen uniformly at random, and $(ii)$ the first player to speaks (namely Alice for odd $r$ and Bob for even $r$) is \emph{oblivious} to the identity of $\zstar$. 

The proof of the communication lower bound then goes as follows. Using information-theoretic arguments, we can argue that if the first message of the protocol is of size 
$o(n^{1/r})$, then it only reveals $o(1)$ bits of information about an ``average'' embedded instance $I_i$ for $i \in [n_r]$ of $\dist_{r-1}$. In particular, since the sender of the first message is oblivious to the identity of the $\zstar$ (by property $(ii)$), the first 
message only reveals $o(1)$ bits of information about the instance $I_{\zstar}$. This effectively means that the distribution of the instance $I_{\zstar}$ is essentially the same as $\dist_{r-1}$ even after the first round. 
However, by property $(i)$, the players now need to solve the instance $I_{\zstar}$ on $n^{(r-1)/r}$ 
elements sampled from distribution $\dist_{r-1}$ in $r-1$ rounds. By induction, this requires $\Omega(\paren{n^{(r-1)/r}}^{1/r-1}) = \Omega(n^{1/r})$ bits, which implies
the desired lower bound for $r$-round protocols.  

The outline above is arguably the most straightforward application of round-elimination (see, e.g. the tree-pointer-jumping problem in~\cite{ChakrabartiCM08}). 
Unfortunately however, this approach does not work directly in our application. In particular, in the discussion above, we left the specifics of 
how the $(r-1)$-round instances $I_1,\ldots,I_{n^{1/r}}$ are embedded together to form $I$. For the above information-theoretic arguments to work, these instances need to be sampled \emph{independently} of each other.  On the other hand, for us to be able to 
embed them together in a valid instance of $\TCI$, we need to ensure that they collectively preserve monotonicity and convexity properties of $\TCI$. This requires \emph{correlating} the instances $I_1,\ldots,I_{n^{1/r}}$, impeding
the use of previous information-theoretic argument. 

We get around this challenge by carefully ``revealing extra information'' about the inputs of the players to each other (similar to the reduction from \AugIndex in Lemma~\ref{lem:cc-tci-1-round}), which allows to 
``control'' the correlation between different instances $I_1,\ldots,I_{n^{1/r}}$ in terms of these revealed information. We then show that even with this extra information, the two properties above for embedded instances continue to hold, and at the same time, 
we have enough independence in the instances to make the information-theoretic arguments outlined above work. 

We comment that this construction of hard instances of $\TCI$ and the proof of 
the corresponding communication lower bound is one of the main technical contributions of this paper.

\subsubsection{General Lower Bound: The Hard Input Distribution}\label{sec:cc-tci-general} %

We use an integer $N \geq 1$ as a parameter in defining all other parameters of our hard distribution. In particular, for $r$-round instances, $n_r = N^r$ is the number of points given to Alice and Bob, and 
$m_r := N$ is the number of $(r-1)$-round instances ``embedded'' inside the $r$-round instance. We also define the following two \emph{operators} on instances that are used in our lower bound construction 
(their roles will become more evident once we give the proper definition of the hard distribution). 

\begin{itemize}
	\item \textbf{Slope-Shift Operator:} In any instance $I$ of our hard distribution $\dist_r$, the input to Alice is constructed using several (potentially different) $\Stepcurve$ functions. By applying slope-shift operator on instance $I$ with parameter $\alpha$, 
	we increase the second parameter in every application of $\Stepcurve$ in constructing Alice's input by an \emph{additive} factor of $\alpha$. As a result, any segment in Alice's input constructed with $\Stepcurve(*,\beta)$ 
	becomes $\Stepcurve(*,\alpha+\beta)$. 
	We ensure that the operator also changes the slope of Bob's input by $\alpha$.
	
	\item \textbf{Origin-Shift Operator:} By applying the origin-shift operator with point $p_A \in \IR^2$ {from Alice's side} in an instance $I$ of $\dist_r$, 
	we shift \emph{all} points in the instance $I$ along the same line so that the \emph{left-most} point of Alice's input will be on the point $p_A$. Similarly, by applying the origin-shift operator with $p_B \in \IR^2$  {from Bob's side}, we
	shift \emph{all} points along the same line so that the \emph{right-most} point of Bob's input will be on $p_B$. This operator clearly does not change the slope of any line segment in players' inputs. 
\end{itemize}

\noindent
We are now ready to describe our hard input distribution. 

\paragraph{Distribution $\dist_r$: The Hard Distribution for $r$-round Protocols of $\TCI$.} We define the following procedure $\Instance$ that given a parameter $r$, construct an instance of $\TCI$. 

\begin{tbox}
	$\underline{\Instance(r)}$. 
	\begin{enumerate}
		\item If $r=1$, sample $(A,B)$ from the distribution of Lemma~\ref{lem:cc-tci-1-round}; otherwise, define $(A,B):= \EvenI(r)$ for even $r$ and $(A,B) :=\OddI(r)$ for odd $r$. 
		\item Return the points $(A,B)$ as the $r$-round instance. 
	\end{enumerate}
\end{tbox}


We now define the $\EvenI$ procedure inside $\Instance$. 

\begin{tbox}
	$\underline{\EvenI(r)}$. 
	\begin{enumerate}
		\item Sample $m_r$ instances $(C_i,D_i)$ \emph{independently} from $\Instance(r-1)$. 
		
		\item For $i=m_r$ down to $1$ do: 
		\begin{enumerate}
			\item Let $p^{i+1}_B$ be the left-most point of Bob's input in $(C_{i+1},D_{i+1})$ (define $p^{m_r+1}_B := (n_r,0)$).  
			Apply the origin-shift operator with point $p^{i+1}_B$ from Bob's side on the instance $(C_i,D_i)$. 
			\item Let $\alpha^{i+1}_r$ be the largest slope of any segment in $(C_{i+1},D_{i+1})$. 
			Apply the slope-shift operator with slope $\alpha^{i+1}_r$ on $(C_i,D_i)$. 
		\end{enumerate}
		\item Sample $\zstar_r \in [m_r]$ uniformly at random.
		\item Define $A := (A_1,\ldots,A_{m_r})$ where $A_{\zstar_r} = C_{\zstar_r}$; the remaining $A_{i}$'s for $i \neq \zstar_r$ are constructed by extending the curve in $A_{\zstar_r}$ on both its endpoints along straight lines. 
		\item Define $B := (B_1,\ldots,B_{m_r})$ where $B_i = D_i$ for all $i \in [m_r]$. 
	\end{enumerate}
\end{tbox}

We refer to instances $(C_1,D_1),\ldots,(C_{m_r},D_{m_r})$ as \emph{sub-instances}. Several remarks are in order about these sub-instances. 
Firstly, even though they were originally sampled independently, by applying the origin-shit and slope-shift operators, we have correlated these instances. In particular, each instance $(C^{i},D^{i})$ depends on instances $(C^{j},D^{j})$ for $j > i$. 
Moreover, note that \emph{not} all the points in these instances 
appear in the final instance $(A,B)$. In particular, we only use the points in $C_{\zstar_r}$ to define $A_{\zstar_r}$; the remaining points in $A \setminus A_{\zstar_r}$ are obtained differently from $C_1,\ldots,C_{m_r}$ (the points in $B$ are however
identical to the points in $D_1,\ldots,D_{m_r}$). Nevertheless, the remaining instances still play a marginal role in the definition of the players' inputs because these points define
the starting point and starting slope of each sub-instance. In the following, we refer to $(A,B)$ as the \emph{actual} input of Alice and Bob, and refer to the points in $(C_1,D_1),\ldots,(C_{m_r},D_{m_r})$ that are not part of $(A,B)$ as \emph{fooling} inputs. 
Figure~\ref{fig:EvenI} gives an illustration of $\EvenI$. 

\input{Figs/Instance}

We use the term sub-instance for both $(A_i,B_i)$ and $(C_i,D_i)$ pairs. 
For any $i \in [m_r]$, we use $A_i:= (a_{i,1},\ldots,a_{i,n_{r-1}})$ and $B_i:= (b_{i,1},\ldots,b_{i,n_{r-1}})$ to denote the points in sub-instance $(A_i,B_i)$. 
The following proposition ensures that instances sampled by $\EvenI$ do not violate the monotonicity and convexity properties of $\TCI$ across sub-instances. 

\begin{proposition}\label{prop:EvenI-properties}
	For $(A,B)$ sampled from $\EvenI$, assuming each sub-instance $(A_i,B_i)$ satisfies monotonicity and convexity  of $\TCI$, then $(A,B)$ also satisfies monotonicity and convexity. 
\end{proposition}
\begin{proof}
	For $B_i$'s, the monotonicity and convexity follow from the origin-shift operator and slope-shift operator, respectively. For $A_i$'s, different sub-instances are obtained by extending two line segments in $A_{\zstar_r}$, and hence $A$ trivially
	satisfies the properties.  
\end{proof}

We refer to the instance $(A_{\zstar_r},B_{\zstar_r}) = (C_{\zstar_r},D_{\zstar_r})$ as the \emph{special} sub-instance of $(A,B)$. 
The next proposition signifies the role of the special sub-instance in $\EvenI$.

\begin{proposition}\label{prop:EvenI-special}
	For instances $(A,B)$ sampled from $\EvenI$, the answer to $\TCI(A,B)$ is the same as the answer to $\TCI(C_{\zstar_r},D_{\zstar_r})$. 
\end{proposition}
\begin{proof}
	Since $(A_{\zstar_r},B_{\zstar_r}) = (C_{\zstar_r},D_{\zstar_r})$, and $(C_{\zstar_r},D_{\zstar_r})$ form a valid instance of $\TCI$,  clearly $A$ and $B$ also only cross each other between the points in $(A_{\zstar_r},B_{\zstar_r})$. 
\end{proof}

We now turn to the definition of the $\OddI$ procedure inside $\Instance$.
The definition of $\OddI$ procedure is similar to $\EvenI$ by switching the role of Alice and Bob.  

\begin{tbox}
	$\underline{\OddI(r)}$.
	
	\begin{enumerate}
		\item Sample $m_r$ instances $(C_i,D_i)$ \emph{independently} from $\Instance(r-1)$. 
		
		\item For $i=1$  to $m_r$ do: 
		\begin{enumerate}
			\item Let $p^{i-1}_A$ be the right-most point of Alice's input in $(C_{i-1},D_{i-1})$ (define $p^{0}_A := (0,0)$).  
			Apply the origin-shift operator with point $p^{i-1}_A$ from Alice's side on  instance $(C_i,D_i)$. 
			\item Let $\alpha^{i-1}_r$ be the largest slope of any segment in $(C_{i-1},D_{i-1})$. 
			Apply the slope-shift operator with slope $\alpha^{i-1}_r$ on $(C_i,D_i)$. 
		\end{enumerate}
		\item Sample $\zstar_r \in [m_r]$ uniformly at random.
		\item Define $A := (C_1,\ldots,C_{m_r})$.
		\item Define $B := (B_1,\ldots,B_{m_r})$ where $B_{\zstar_r} := D_{\zstar_r}$; the remaining $B_{i}$'s for $i \neq \zstar_r$ are constructed by extending the curve in $B_{\zstar_r}$ on both its endpoints along straight lines. 
	\end{enumerate}
\end{tbox}

Similar to $\EvenI$, instances of $\OddI$ also consists of $m_r$ sub-instances among which $(C_{\zstar_r},D_{\zstar_r})$ is called the \emph{special} sub-instance. 
Figure~\ref{fig:OddI} gives an illustration of instances sampled by $\OddI$. 

The following two properties are analogous to Propositions~\ref{prop:EvenI-properties} and~\ref{prop:EvenI-special} for $\EvenI$. 

\begin{proposition}\label{prop:OddI-properties}
	For $(A,B)$ sampled from $\OddI$, assuming each sub-instance $(A_i,B_i)$ satisfies monotonicity and convexity  of $\TCI$, then $(A,B)$ also satisfies monotonicity and convexity. 
\end{proposition}

\begin{proposition}\label{prop:OddI-special}
	For instances $(A,B)$ sampled from $\OddI$, the answer to $\TCI(A,B)$ is the same as the answer to $\TCI(C_{\zstar_r},D_{\zstar_r})$. 
\end{proposition}

\paragraph{Actual vs Fooling Inputs.} As we already observed in the proof of Lemma~\ref{lem:cc-tci-1-round}, providing the players with extra 
information about the input of the other player (i.e., giving Bob the first $\istar$ points in Alice's input) facilitates the proof of the lower bound. This is also the case for our hard instances for $r>1$ round protocols. 
In the following observations, we state several properties of this extra information which is crucial for our information-theoretic lower bound for $\TCI$.

\begin{observation}\label{obs:construct}
	In $\EvenI$, there is a one-to-one mapping between $(A_{\zstar_r},B_{\zstar_r})$ and the {original} $(C_{\zstar_r},D_{\zstar_r})$ ({before applying any operator}),  
	assuming we are given $(C^{>\zstar_r},D^{>\zstar_r})$.
	Similarly,in $\OddI$, there is a one-to-one mapping between $(A_{\zstar_r},B_{\zstar_r})$ and the original sub-instance
	$(C_{\zstar_r},D_{\zstar_r})$ (before applying any operator), assuming we are given $(C^{<\zstar_r},D^{<\zstar_r})$
\end{observation}

The reason behind Observation~\ref{obs:construct} is simply the operators applied to each $(C_i,D_i)$ are functions of $(C^{>i},D^{>i})$ in $\EvenI$ (resp. $(C^{<i},D^{<i})$ in $\OddI$) and the special sub-instance 
is just a ``copy'' of $(C_{\zstar_r},D_{\zstar_r})$. 

Observation~\ref{obs:construct} implies that if players have access to $(C^{>\zstar_r},D^{>\zstar_r})$ in $\EvenI$ (resp. $(C^{<\zstar_r},D^{<\zstar_r})$ in $\OddI$) as an extra input, 
then they can determine the original distribution of their special sub-instance. This is the main reason that we provide the players with this extra input in our reduction.

\begin{observation}\label{obs:hidden-special}
	In $\EvenI$, the index $\zstar_r \in [m_r]$ is chosen independently of $B = (B_1,\ldots,B_{m_r})$ and $(C_1,D_1),\ldots,(C_{m_r},D_{m_r})$. Similarly, in  $\OddI$, 
	the index $\zstar_r \in [m_r]$ is chosen independently of $A = (A_1,\ldots,A_{m_r})$ and $(C_1,D_1),\ldots,(C_{m_r},D_{m_r})$. 
\end{observation}

This observation follows directly from the construction of the instances. 
Observation~\ref{obs:hidden-special} implies that \emph{even given the extra input}, the player that sends the first message
is oblivious to the identity of the special sub-instance.

\subsubsection{General Lower Bound: The Communication Complexity}\label{sec:cc-tci-lower}

We prove Theorem~\ref{thm:cc-tci} by induction on the number of rounds, with Lemma~\ref{lem:cc-tci-1-round} forming the base of the induction. 
We have the following lemma.

\begin{lemma}\label{lem:cc-lb}
	For any $r \geq 1$ and any $(1/3)$-error protocol $\prot_r$ for instances of $\TCI$ sampled from the distribution $\dist_r$, the communication cost of $\prot_r$ is $\Omega(N/r^2)$. 
\end{lemma}

From now on we fix a \emph{deterministic} protocol $\prot_r$ for $\TCI$ on $\dist_r$; we later use Yao's minimax principle~\cite{Yao83} to extend the lower bound to randomized protocols. We use $\rProt$ to 
denote the random variable for messages communicated in the protocol, and write $\rProt = (\rProt_1,\ldots,\rProt_r)$, where $\rProt_\ell$ denotes the message communicated in round $\ell$.  
We further use $\rZ$ to denote the index $\zstar_r$, which corresponds to the index of the special sub-instance. Let $(\rA_1,\ldots,\rA_{m_r})$ and $(\rB_1,\ldots,\rB_{m_r})$ denote the random variables 
for the points $A$ and $B$ and their partitioning into sub-instances respectively, and $(\rC_1,\ldots,\rC_{m_r})$ and $(\rD_1,\ldots,\rD_{m_r})$ for $C$ and $D$.  

We start with the following lemma that formalizes our intuition that players cannot reveal information about the special sub-instance in their 
first round. We first consider even-round protocols.

\begin{lemma}\label{lem:ds-e}
	For any even integer $r$, and any $r$-round protocol $\prot_r$ with worst-case message length $\ell$ on instances of $\dist_r$, 
	\[\mi{(\rA_{\rZ},\rB_{\rZ})}{\rProt_1 \mid \rC^{>\rZ},\rD^{>\rZ},\rZ} \leq \ell/N.\]
\end{lemma}
\begin{proof}
	We start by expanding the LHS:
	\begin{align*}
	&\mi{(\rA_{\rZ},\rB_{\rZ})}{\rProt_1 \mid \rC^{>\rZ},\rD^{>\rZ},\rZ} \\
	&\hspace{1cm}= \Ex_{z \in [m_r]} \bracket{\mi{(\rA_{z},\rB_{z})}{\rProt_1 \mid \rC^{>z},\rD^{>z},\rZ=z}} \tag{definition of conditional mutual information} \\
	&\hspace{1cm}= \frac{1}{m_r} \sum_{z=1}^{m_r}\mi{(\rA_{z},\rB_{z})}{\rProt_1 \mid \rC^{>z},\rD^{>z},\rZ=z} \tag{distribution of $\rZ$ is uniform over $[m_r]$} \\
	&\hspace{1cm}= \frac{1}{m_r} \sum_{z=1}^{m_r}\mi{(\rC_{z},\rD_{z})}{\rProt_1 \mid \rC^{>z},\rD^{>z},\rZ=z} \tag{by Observation~\ref{obs:construct}} \\
	&\hspace{1cm}=  \frac{1}{m_r} \sum_{z=1}^{m_r}\mi{(\rC_{z},\rD_{z})}{\rProt_1 \mid \rC^{>z},\rD^{>z}},
	\end{align*}
	where the last equality is due to the fact that the joint distribution of all random variables $(\rC_z,\rD_z), \rProt_r, \rC^{>z},\rD^{>z}$ is independent of the event $\rZ=z$. Indeed, for even $r$, the message $\rProt_1$ sent by Bob is a function
	of $(B_1,\ldots,B_{m_r})$ and $(C_1,D_1), \ldots,(C_{m_r},D_{m_r})$, and by Observation~\ref{obs:hidden-special}, these random variables are all independent of $\zstar_r$. As such, removing the conditioning on the event $\rZ=z$ does not
	change the distribution of variables above. Finally, 
	\begin{align*}
	&\frac{1}{m_r} \sum_{z=1}^{m_r}\mi{(\rC_{z},\rD_{z})}{\rProt_r \mid \rC^{>z},\rD^{>z}} \\
	&\hspace{1cm}= \frac{1}{m_r} \cdot \mi{\rC_{1},\ldots,\rC_{m_r},\rD_{1},\ldots,\rD_{m_r}}{\rProt_1} \tag{chain rule of mutual information (\itfacts{chain-rule})} \\
	&\hspace{1cm}\leq \frac{1}{m_r} \cdot \en{\rProt_1} \leq \frac{\ell}{m_r} \tag{by~\itfacts{uniform}, $\en{\rProt_1} \leq \ell$}.
	\end{align*}
	The lemma follows by noting that $m_r = N$.
\end{proof}

The following lemma for odd-round protocols is analogous to Lemma~\ref{lem:ds-e} for even-round ones (but note the change in the order of conditioning). 

\begin{lemma}\label{lem:ds-o}
	For any odd integer $r$ and any $r$-round protocol $\prot_r$ with worst-case message length $\ell$ on instances of $\dist_r$, 
	\[\mi{(\rA_{\rZ},\rB_{\rZ})}{\rProt_1 \mid \rC^{<\rZ},\rD^{<\rZ},\rZ} \leq \ell/N.\]
\end{lemma}

To continue, we need the following definition. 
\begin{itemize}
	\item[] \underline{Distribution $\mu_e$ for even $r$}: For an assignment $(\Prot_1,C^{>z},D^{>z},z)$ (denoted by $E$ for short) to $(\rProt_1,\rC^{>\rZ},\rD^{>\rZ},\rZ)$, 
	we define $\mu_e(E)$ as the distribution of $(\rA_{z},\rB_z)$ in $\dist_r$ conditioned on $\rProt_r = \Prot_r, \rZ = z, \rC^{>z} = C^{>z},\rD^{>z}=D^{>z}$.
	\item[] \underline{Distribution $\mu_o$ for odd $r$}: For an assignment $(\Prot_r,C^{<z},D^{<z},z)$ (denoted by $O$ for short) to $(\rProt_1,\rC^{<\rZ},\rD^{<\rZ},\rZ)$, we define $\mu_o(O)$ as the distribution of $(\rA_{z},\rB_z)$ in $\dist_r$ conditioned
	on $\rProt_r = \Prot_r, \rZ = z, \rC^{<z} = C^{<z},\rD^{<z}=D^{<z}$.
\end{itemize}

Using Lemma~\ref{lem:ds-e}, we have the following claim.
\begin{claim}\label{clm:dist-close-e}
	For any even integer $r$ and any $r$-round protocol $\prot_r$ with worst-case message length $o(N/r^2)$,  
	\[
	\Ex_{E=(\Prot_1,C^{>z},D^{>z},z)}\bracket{\tvd{\mu_e(E)}{\dist_{r-1}}} = o(1/r).
	\]
\end{claim}
\begin{proof}
	By the connection between mutual information 
	and KL-divergence (Fact~\ref{fact:kl-info}), we have, 
	\begin{align*}
	&\mi{(\rA_{\rZ},\rB_{\rZ})}{\rProt_1 \mid \rC^{>\rZ},\rD^{>\rZ},\rZ} \\
	&\hspace{0.5cm}= \Ex_{E=(\Prot_1,C^{>z},D^{>z},z)}\bracket{\kl{\distribution{\rA_{\rZ},\rB_{\rZ} \mid E \setminus\Prot_1}}{\distribution{\rA_{\rZ},\rB_{\rZ} \mid E}}} \\
	&\hspace{0.5cm}= \Ex_{E=(\Prot_1,C^{>z},D^{>z},z)}\bracket{\kl{\dist_{r-1}}{\mu_e(E)}};
	\end{align*}
	Conditioned on $\rZ=z$, distribution of $(\rA_{\rZ},\rB_{\rZ})$ is the same as the original distribution of $(\rC_{\rZ},\rD_{\rZ})$ by Observation~\ref{obs:construct}. Furthermore, 
	\begin{align*}
	&\Ex_{E=(\Prot_1,C^{>z},D^{>z},z)}\bracket{\kl{\dist_{r-1}}{\mu_e(E)}} \\
	&\hspace{0.5cm}\geq \Ex_{E=(\Prot_r,C^{>z},D^{>z},z)}\bracket{2 \cdot \tvd{\dist_{r-1}}{\mu_e(E)}^2} \tag{Pinsker's inequality (Fact~\ref{fact:pinskers})}\\
	&\hspace{0.5cm}\geq 2 \cdot \paren{\Ex_{E=(\Prot_r,C^{>z},D^{>z},z)}\bracket{\tvd{\dist_{r-1}}{\mu_e(E)}}}^2 \tag{Jensen's inequality}.
	\end{align*}
	By Lemma~\ref{lem:ds-e}, $\mi{(\rA_{\rZ},\rB_{\rZ})}{\rProt_r \mid \rC^{>\rZ},\rD^{>\rZ},\rZ} = o(1/r^2)$, implying that,
	\begin{align*}
	\Ex_{E=(\Prot_r,C^{>z},D^{>z},z)}\bracket{\tvd{\dist_{r-1}}{\mu_e(E)}} = o(1/r).
	\end{align*}
	This finalizes the proof. 
\end{proof}

The following claim for odd-round protocols is analogous to Claim~\ref{clm:dist-close-e} for even-round ones (again note the change in the order of conditioning and the distribution). 

\begin{claim}\label{clm:dist-close-o}
	For any odd integer $r$ and any $r$-round protocol $\prot_r$ with worst-case message length $o(N/r^2)$,  
	\[
	\Ex_{O=(\Prot_1,C^{<z},D^{<z},z)}\bracket{\tvd{\mu_o(O)}{\dist_{r-1}}} = o(1/r).
	\]
\end{claim}

Define the recursive function $\delta(k) = \delta(k-1) - o(1/r)$ with base case $\delta(1) = 1/4$ (here $r$ is the number of rounds). 

\begin{lemma}\label{lem:cc-tci-induction}
	Any deterministic $\delta(r)$-error protocol $\prot_r$ on $\dist_r$ requires $\Omega(N/r^2)$ communication. 
\end{lemma}
\begin{proof}
	The proof is by induction on the number of rounds. The base case for $r=1$ follows from Lemma~\ref{lem:cc-tci-1-round}.
	We now prove the induction step. 
	
	Suppose the lemma holds for all integers up to $r - 1$, we prove it for $r$-round protocols. Given a $r$-round protocol $\prot_r$ for $\dist_r$ that violates the 
	induction hypothesis, we construct a $(r-1)$-round protocol $\prot_{r-1}$ for $\dist_{r-1}$ that also violates the induction hypothesis, a contradiction. The protocol $\prot_{r-1}$ is constructed in two steps: we first construct a 
	randomized protocol $\prot'$ from $\prot_{r}$, and then fix the randomness of the protocol to achieve a deterministic protocol. 
	
	We now describe $\prot'$. For simplicity, we only give the protocol  for  even choices of $r$; the extension to 
	odd values is straightforward.  Given an instance $(A,B) \sim \dist_{r-1}$, protocol $\prot'$ works as follows: 
	\begin{tbox}
		\begin{enumerate}
			\item Using \underline{public randomness}, the players sample $(\Prot_1,C^{>z},D^{>z},\zstar_r)$ from the distribution $\dist_r$.
			\item Bob samples remaining coordinates $(C_j,D_j)$ for $j < \zstar_r$ using \underline{private randomness} from distribution $\dist_r \mid (\Prot_1,C^{>z},D^{>z},\zstar_r)$. 
			\item Alice sets $A_{\zstar_r} = A$ and Bob sets $B_{\zstar_r} = B$ by applying the appropriate slope-shift and origin-shift operators based on $C^{>z},D^{>z}$ which is known to both Alice and Bob. 
			\item The players then fill the rest of their input in $(A_1,\ldots,A_{m_r})$ and $(B_1,\ldots,B_{m_r})$; Bob knows all of $(C_1,D_1),\ldots,(C_{m_r},D_{m_r})$ and can perform the needed slope-shift and origin-shift operators, and Alice simply needs to 
			extend $A_{\zstar_r}$ across straight lines. 
			\item The players run $\prot_r$ on these new points from the second round onwards, assuming that the first communicated message was $\Prot_1$. They output the index returned by $\prot_r$.
		\end{enumerate}
	\end{tbox}
	Communication cost of $\prot'$ is clearly at most as the communication cost of $\prot_{r}$. We now prove the correctness of $\prot'$. 
	\begin{claim}\label{clm:prot'-correct}
		Assuming $\prot_r$ is a $\delta(r)$-error protocol for $\dist_r$, $\prot'$ will be a $(\delta(r)+o(1/r))$-error protocol for $\dist_{r-1}$. 
	\end{claim}
	\begin{proof}
		We have,
		\begin{align*}
		& \Pr_{\dist_{r-1}}\paren{\prot'~\errs} \\
		&= \Ex_{E=(\Prot_1,C^{>z},D^{>z},\zstar_r)}\bracket{\Pr_{\dist_{r-1}}\paren{\prot_r~\errs \mid E}} \tag{by Proposition~\ref{prop:EvenI-special}} \\
		&\leq \Ex_{E=(\Prot_1,C^{>z},D^{>z},\zstar_r)}\bracket{\Pr_{\mu_e(E)}\paren{\prot_r~\errs} + \tvd{\mu_e(E)}{\dist_{r-1}}} \tag{by Fact~\ref{fact:tvd-small}}\\
		&= \Pr_{\dist_r}\paren{\prot_r~\errs} + \Ex_{E=(\Prot_1,C^{>z},D^{>z},\zstar_r)}\bracket{\tvd{\mu_e(E)}{\dist_{r-1}}} \tag{by linearity of expectation} \\
		&\leq \delta(r) + o(1/r^2) \tag{$\prot_r$ is a $(\delta(r))$-error protocol, and by Claim~\ref{clm:dist-close-e}},
		\end{align*}
		finalizing the proof. 		
	\end{proof}
	
	We are now ready to complete the proof of Lemma~\ref{lem:cc-tci-induction}. By Claim~\ref{clm:prot'-correct}, $\prot'$ is a $(\delta(r)+o(1/r))$-error protocol for $\dist_{r-1}$. $\prot'$ is a randomized protocol. 
	However, by an averaging argument, we can fix the randomness of $\prot'$ to obtain a deterministic  $(\delta(r)+o(1/r))$-error protocol for $\dist_{r-1}$ with the same communication cost $o(N/r^2)$. As 
	$\delta(r) + o(1/r) = \delta(r-1)$, this contradicts the induction hypothesis. We thus have that the communication cost of $\prot_r$ is $\Omega(N/r^2)$, proving the induction step. This concludes the proof.
\end{proof}

Lemma~\ref{lem:cc-lb} now follows immediately from Lemma~\ref{lem:cc-tci-induction} as $\delta(r) = 1/4 + \sum_{k=1}^{r}o(1/r) = 1/4 + o(1) < 1/3$, and by the easy direction of Yao's minimax principle~\cite{Yao83}.

\subsubsection{Proof of Theorem~\ref{thm:cc-tci}}\label{sec:cc-tci-proof}


\begin{proof}[Proof of Theorem~\ref{thm:cc-tci}]
	By Lemma~\ref{lem:cc-lb}, any $(1/3)$-error $r$-round protocol for $\TCI_n$ requires $\Omega(N/r^2)$ communication on instances of $\dist_r$. In these instances, $n=N^r$ by the construction of $\dist_r$. Plugging in $N = n^{1/r}$, we obtain 
	$\CCr{\TCI_n}{r} = \Omega(\frac{1}{r^2} \cdot n^{1/r})$.  
	
	We conclude this proof by making the following remark: A $r$-round instance of our problem consists of at most $N^{r-1}$ applications of $\Stepcurve$, each having a larger slope than the 
	previous one by an additive factor of $N$. As a result, the largest slope using in our construction is $N^{O(r)}$. This implies that the bit-complexity of the numbers we use is bounded by $\log{(N^{O(r)})}  = O(\log{n})$. 
\end{proof}

As a corollary of Theorem~\ref{thm:cc-tci}, using the connection between two-curve intersection problem and linear programming outlined in Section~\ref{sec:tci-problem}, we obtain the following. 

\begin{corollary}\label{cor:LP-cc}
	For any integer $r \geq 1$, any two-player $r$-round protocol for $2$-dimensional linear programming with $n$ constraints requires $\Omega(\frac{1}{r^2} \cdot n^{1/r})$ communication. 
\end{corollary}

%% file: Figs/Instance.tex
\begin{figure}[t]
    \centering
    \subcaptionbox{An Illustration of $\EvenI$. \label{fig:EvenI}}[0.45\textwidth]{

\begin{tikzpicture}[ auto ,node distance =1cm and 2cm , on grid , semithick , state/.style ={ circle ,top color =white , bottom color = white , draw, black , text=black}, every node/.style={inner sep=0,outer sep=0}]

\draw[help lines,xstep=0.5,ystep=.25] (0,0) grid (8,4);

\node[state, circle, ForestGreen, line width=0.35mm, minimum height=3pt, minimum width=3pt] (b1) at (8,0){};
\node[state, circle, ForestGreen, line width=0.35mm, minimum height=3pt, minimum width=3pt] (b2) at (6,0.5){};
\node[state, circle, ForestGreen, line width=0.35mm, minimum height=3pt, minimum width=3pt] (b3) at (4,1.25){};
\node[state, circle, ForestGreen, line width=0.35mm, minimum height=3pt, minimum width=3pt] (b4) at (2,2.25){};
\node[state, circle, ForestGreen, line width=0.35mm, minimum height=3pt, minimum width=3pt] (b5) at (0,4){};

\draw[ForestGreen, line width=0.5mm] (b1) -- (b2);
\draw[ForestGreen, line width=0.5mm] (b2) -- (b3);
\draw[ForestGreen, line width=0.5mm] (b3) -- (b4);
\draw[ForestGreen, line width=0.5mm] (b4) -- (b5);



\node[state, circle, blue!25, line width=0.35mm, minimum height=3pt, minimum width=3pt] (a11) at (6,0){};
\node[state, circle, blue!25, line width=0.35mm, minimum height=3pt, minimum width=3pt] (a12) at (6.5,0){};
\node[state, circle, blue!25, line width=0.35mm, minimum height=3pt, minimum width=3pt] (a13) at (7,0.25){};
\node[state, circle, blue!25, line width=0.35mm, minimum height=3pt, minimum width=3pt] (a14) at (7.5,0.75){};
\node[state, circle, blue!25, line width=0.35mm, minimum height=3pt, minimum width=3pt] (a15) at (8,1.75){};

\draw[dashed, red!50, line width=0.25mm] (a11) rectangle (a15);

\draw[blue!25, line width=0.25mm] (a11) -- (a12);
\draw[blue!25, line width=0.25mm] (a12) -- (a13);
\draw[blue!25, line width=0.25mm] (a13) -- (a14);
\draw[blue!25, line width=0.25mm] (a14) -- (a15);


\node[state, circle, blue!25, line width=0.35mm, minimum height=3pt, minimum width=3pt] (a21) at (4,0.5){};
\node[state, circle, blue!25, line width=0.35mm, minimum height=3pt, minimum width=3pt] (a22) at (4.5,0.75){};
\node[state, circle, blue!25, line width=0.35mm, minimum height=3pt, minimum width=3pt] (a23) at (5,1){};
\node[state, circle, blue!25, line width=0.35mm, minimum height=3pt, minimum width=3pt] (a24) at (5.5,1.5){};
\node[state, circle, blue!25, line width=0.35mm, minimum height=3pt, minimum width=3pt] (a25) at (6,3){};

\draw[dashed, red!50, line width=0.25mm] (a21) rectangle (a25);

\draw[blue!25, line width=0.25mm] (a21) -- (a22);
\draw[blue!25, line width=0.25mm] (a22) -- (a23);
\draw[blue!25, line width=0.25mm] (a23) -- (a24);
\draw[blue!25, line width=0.25mm] (a24) -- (a25);


\node[state, circle, blue, line width=0.35mm, minimum height=3pt, minimum width=3pt] (a30) at (0,0.25){};
\node[state, circle, blue, line width=0.35mm, minimum height=3pt, minimum width=3pt] (a31) at (2,1.25){};
\node[state, circle, blue, line width=0.35mm, minimum height=3pt, minimum width=3pt] (a32) at (2.5,1.5){};
\node[state, circle, blue, line width=0.35mm, minimum height=3pt, minimum width=3pt] (a33) at (3,2){};
\node[state, circle, blue, line width=0.35mm, minimum height=3pt, minimum width=3pt] (a34) at (3.5,3){};
\node[state, circle, blue, line width=0.35mm, minimum height=3pt, minimum width=3pt] (a35) at (4,4){};

\draw[dashed, red!50, line width=0.25mm] (a31) rectangle (a35);

\draw[blue, line width=0.5mm] (a31) -- (a30);
\draw[blue, line width=0.5mm] (a31) -- (a32);
\draw[blue, line width=0.5mm] (a32) -- (a33);
\draw[blue, line width=0.5mm] (a33) -- (a34);
\draw[blue, line width=0.5mm] (a34) -- (a35);


\node[state, circle, blue!25, line width=0.35mm, minimum height=3pt, minimum width=3pt] (a41) at (0,2.25){};
\node[state, circle, blue!25, line width=0.35mm, minimum height=3pt, minimum width=3pt] (a42) at (0.5,2.5){};
\node[state, circle, blue!25, line width=0.35mm, minimum height=3pt, minimum width=3pt] (a43) at (1,2.75){};
\node[state, circle, blue!25, line width=0.35mm, minimum height=3pt, minimum width=3pt] (a44) at (1.5,3.25){};
\node[state, circle, blue!25, line width=0.35mm, minimum height=3pt, minimum width=3pt] (a45) at (2,4){};

\draw[dashed, red!50, line width=0.25mm] (a41) rectangle (a45);

\draw[blue!25, line width=0.25mm] (a41) -- (a42);
\draw[blue!25, line width=0.25mm] (a42) -- (a43);
\draw[blue!25, line width=0.25mm] (a43) -- (a44);
\draw[blue!25, line width=0.25mm] (a44) -- (a45);

\end{tikzpicture}

}
\hspace{0.75cm}   
    \subcaptionbox{An Illustration of $\OddI$. \label{fig:OddI}}[0.45\textwidth]{

\begin{tikzpicture}[ auto ,node distance =1cm and 2cm , on grid , semithick , state/.style ={ circle ,top color =white , bottom color = white , draw, black , text=black}, every node/.style={inner sep=0,outer sep=0}]

\draw[help lines,xstep=0.5,ystep=.25] (0,0) grid (8,4);

\node[state, circle, blue, line width=0.35mm, minimum height=3pt, minimum width=3pt] (a1) at (0,0){};
\node[state, circle, blue, line width=0.35mm, minimum height=3pt, minimum width=3pt] (a2) at (1,0.25){};
\node[state, circle, blue, line width=0.35mm, minimum height=3pt, minimum width=3pt] (a3) at (2,0.5){};
\node[state, circle, blue, line width=0.35mm, minimum height=3pt, minimum width=3pt] (a4) at (3,1){};
\node[state, circle, blue, line width=0.35mm, minimum height=3pt, minimum width=3pt] (a5) at (4,1.5){};
\node[state, circle, blue, line width=0.35mm, minimum height=3pt, minimum width=3pt] (a6) at (5,2){};
\node[state, circle, blue, line width=0.35mm, minimum height=3pt, minimum width=3pt] (a7) at (6,2.5){};
\node[state, circle, blue, line width=0.35mm, minimum height=3pt, minimum width=3pt] (a8) at (7,3){};
\node[state, circle, blue, line width=0.35mm, minimum height=3pt, minimum width=3pt] (a9) at (8,4){};

\draw[blue, line width=0.5mm] 
(a1) -- (a2)
(a2) -- (a3)
(a3) -- (a4)
(a4) -- (a5)
(a5) -- (a6)
(a6) -- (a7)
(a7) -- (a8)
 (a8) -- (a9);



\node[state, circle, ForestGreen!50, line width=0.35mm, minimum height=3pt, minimum width=3pt] (b11) at (2,0){};
\node[state, circle, ForestGreen!50, line width=0.35mm, minimum height=3pt, minimum width=3pt] (b12) at (1.5,0){};
\node[state, circle, ForestGreen!50, line width=0.35mm, minimum height=3pt, minimum width=3pt] (b13) at (1,0.25){};
\node[state, circle, ForestGreen!50, line width=0.35mm, minimum height=3pt, minimum width=3pt] (b14) at (0.5,0.5){};
\node[state, circle, ForestGreen!50, line width=0.35mm, minimum height=3pt, minimum width=3pt] (b15) at (0,1){};

\draw[dashed, red!50, line width=0.25mm] (b11) rectangle (b15);

\draw[ForestGreen!50, line width=0.25mm] 
(b11) -- (b12)
(b12) -- (b13)
(b13) -- (b14)
(b14) -- (b15);


\node[state, circle, ForestGreen!50, line width=0.35mm, minimum height=3pt, minimum width=3pt] (b21) at (4,0){};
\node[state, circle, ForestGreen!50, line width=0.35mm, minimum height=3pt, minimum width=3pt] (b22) at (3.5,0.5){};
\node[state, circle, ForestGreen!50, line width=0.35mm, minimum height=3pt, minimum width=3pt] (b23) at (3,1){};
\node[state, circle, ForestGreen!50, line width=0.35mm, minimum height=3pt, minimum width=3pt] (b24) at (2.5,2){};
\node[state, circle, ForestGreen!50, line width=0.35mm, minimum height=3pt, minimum width=3pt] (b25) at (2,4){};

\draw[dashed, red!50, line width=0.25mm] (b21) rectangle (b25);

\draw[ForestGreen!50, line width=0.25mm] 
(b21) -- (b22)
(b22) -- (b23)
(b23) -- (b24)
(b24) -- (b25);


\node[state, circle, ForestGreen, line width=0.35mm, minimum height=3pt, minimum width=3pt] (b30) at (8,0){};
\node[state, circle, ForestGreen, line width=0.35mm, minimum height=3pt, minimum width=3pt] (b31) at (6,.5){};
\node[state, circle, ForestGreen, line width=0.35mm, minimum height=3pt, minimum width=3pt] (b32) at (5.5,1.25){};
\node[state, circle, ForestGreen, line width=0.35mm, minimum height=3pt, minimum width=3pt] (b33) at (5,2){};
\node[state, circle, ForestGreen, line width=0.35mm, minimum height=3pt, minimum width=3pt] (b34) at (4.5,2.75){};
\node[state, circle, ForestGreen, line width=0.35mm, minimum height=3pt, minimum width=3pt] (b35) at (4,4){};

\draw[dashed, red!50, line width=0.25mm] (b31) rectangle (b35);

\draw[ForestGreen, line width=0.5mm] 
(b30) -- (b31)
(b31) -- (b32)
(b32) -- (b33)
(b33) -- (b34)
(b34) -- (b35);


\node[state, circle, ForestGreen!50, line width=0.35mm, minimum height=3pt, minimum width=3pt] (b41) at (8,2){};
\node[state, circle, ForestGreen!50, line width=0.35mm, minimum height=3pt, minimum width=3pt] (b42) at (7.5,2.25){};
\node[state, circle, ForestGreen!50, line width=0.35mm, minimum height=3pt, minimum width=3pt] (b43) at (7,2.75){};
\node[state, circle, ForestGreen!50, line width=0.35mm, minimum height=3pt, minimum width=3pt] (b44) at (6.5,3.25){};
\node[state, circle, ForestGreen!50, line width=0.35mm, minimum height=3pt, minimum width=3pt] (b45) at (6,4){};

\draw[dashed, red!50, line width=0.25mm] (b41) rectangle (b45);

\draw[ForestGreen!50, line width=0.25mm] 
(b41) -- (b42)
(b42) -- (b43)
(b43) -- (b44)
(b44) -- (b45);

\end{tikzpicture}

}
    
\caption{An illustration of $\EvenI$ and $\OddI$. Thick blue and green curves denote the actual inputs of Alice and Bob, respectively. Similarly, light blue and green curves denote the fooling inputs. Each red dashed rectangle denotes one sub-instance.}
\end{figure}
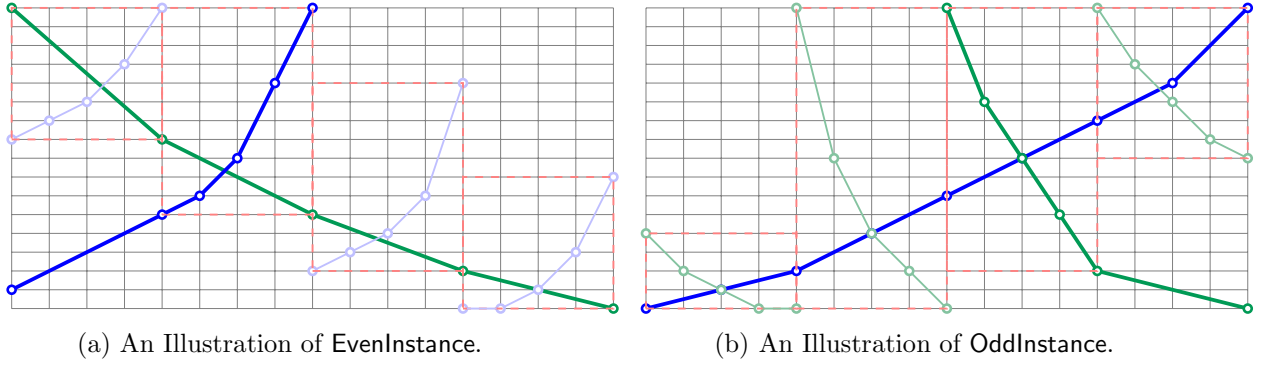

%% file: lb-implications.tex
\subsection{Lower Bounds for Linear Programming in Big Data Models}\label{sec:lb-implication}

We now give some straightforward applications of our communication complexity lower bound for linear programming to streaming and coordinator models, and formalize Result~\ref{res:lower}. 

\paragraph{The Streaming Model.}\label{sec:lower-stream}
It is well-known that communication complexity lower bounds imply space lower bounds on the space complexity of streaming algorithms (see, e.g.~\cite{AMS99,GM08}). Using this connection in conjunction with Corollary~\ref{cor:LP-cc},
we have,  
to establish the following theorem. 

\begin{theorem}\label{thm:lower-stream}
	For any integer $r \geq 1$, any streaming algorithm that makes $r$ passes over the constraints of a $2$-dimensional linear program with $n$ constraints and finds the optimal solution with probability at least $2/3$ 
	requires $\Omega(\frac{1}{r^3} \cdot n^{1/2r})$ space. 
\end{theorem}
\begin{proof}
	It is a standard fact that any streaming algorithm with $p$-passes and $s$-space can be turned into a communication protocol in the two-party communication model with at most $2p$-rounds and $O(p \cdot s)$ communication; see, e.g.~\cite{GM08}.
	The lower bound on the space complexity of the streaming algorithms now follows from Corollary~\ref{cor:LP-cc}. 
\end{proof}

\paragraph{The Coordinator Model.}\label{sec:lower-coordinator}
Any $r$-round distributed protocol implies a $2r$-round protocol in our communication model. Hence, 

\begin{theorem}\label{thm:lower-coordinator}
	For any integer $r \geq 1$, any $r$-round algorithm that finds the optimal solution of a $2$-dimensional linear program with $n$ constraints partitioned across $k \geq 2$ sites in the coordinator model with probability at least $2/3$ requires 
	$\Omega(\frac{1}{r^2} \cdot n^{1/2r})$ communication. 
\end{theorem}
\begin{proof}
	One can turn any $r$-round algorithm in the coordinator model into a $(2r)$-round communication algorithm in the two-party communication model with the same communication cost using the straightforward reduction. 
	The lower bound on the communication cost of  algorithms in the coordinator model now follows from Corollary~\ref{cor:LP-cc}. 
\end{proof}

%% file: ms.bbl
\begin{thebibliography}{10}

\bibitem{AG11}
Kook~Jin Ahn and Sudipto Guha.
\newblock Linear programming in the semi-streaming model with application to
  the maximum matching problem.
\newblock In {\em ICALP 2011}, pages 526--538, 2011.

\bibitem{AMS99}
Noga Alon, Yossi Matias, and Mario Szegedy.
\newblock The space complexity of approximating the frequency moments.
\newblock {\em J. Comput. Syst. Sci.}, 58(1):137--147, 1999.

\bibitem{ACGKOPVW15}
Molham Aref, Balder ten Cate, Todd~J. Green, Benny Kimelfeld, Dan Olteanu, Emir
  Pasalic, Todd~L. Veldhuizen, and Geoffrey Washburn.
\newblock Design and implementation of the logicblox system.
\newblock In {\em SIGMOD}, pages 1371--1382, 2015.

\bibitem{AssadiKL16}
Sepehr Assadi, Sanjeev Khanna, and Yang Li.
\newblock Tight bounds for single-pass streaming complexity of the set cover
  problem.
\newblock In {\em Proceedings of the 48th Annual {ACM} {SIGACT} Symposium on
  Theory of Computing, {STOC} 2016, Cambridge, MA, USA, June 18-21, 2016},
  pages 698--711, 2016.

\bibitem{BKS17}
Paul Beame, Paraschos Koutris, and Dan Suciu.
\newblock Communication steps for parallel query processing.
\newblock {\em J. {ACM}}, 64(6):40:1--40:58, 2017.

\bibitem{Bland78}
Robert~G Bland and Michel Las~Vergnas.
\newblock Orientability of matroids.
\newblock {\em Journal of Combinatorial Theory, Series B}, 24(1):94--123, 1978.

\bibitem{BGV92}
Bernhard~E. Boser, Isabelle Guyon, and Vladimir Vapnik.
\newblock A training algorithm for optimal margin classifiers.
\newblock In {\em COLT}, pages 144--152, 1992.

\bibitem{BCM99}
Herv{\'{e}} Br{\"{o}}nnimann, Bernard Chazelle, and Ji{\v{r}}{\'{\i}}
  Matou{\v{s}}ek.
\newblock Product range spaces, sensitive sampling, and derandomization.
\newblock {\em {SIAM} J. Comput.}, 28(5):1552--1575, 1999.

\bibitem{BG95}
Herv{\'{e}} Br{\"{o}}nnimann and Michael~T. Goodrich.
\newblock Almost optimal set covers in finite vc-dimension.
\newblock {\em Discrete {\&} Computational Geometry}, 14(4):463--479, 1995.

\bibitem{Burges98}
Christopher J.~C. Burges.
\newblock A tutorial on support vector machines for pattern recognition.
\newblock {\em Data Min. Knowl. Discov.}, 2(2):121--167, 1998.

\bibitem{ChakrabartiCM08}
Amit Chakrabarti, Graham Cormode, and Andrew McGregor.
\newblock Robust lower bounds for communication and stream computation.
\newblock In {\em STOC}, pages 641--650, 2008.

\bibitem{Chan16}
Timothy~M. Chan.
\newblock Improved deterministic algorithms for linear programming in low
  dimensions.
\newblock In {\em SODA}, pages 1213--1219, 2016.

\bibitem{CC07}
Timothy~M. Chan and Eric~Y. Chen.
\newblock Multi-pass geometric algorithms.
\newblock {\em Discrete {\&} Computational Geometry}, 37(1):79--102, 2007.

\bibitem{Chao82}
M.~T. Chao.
\newblock A general purpose unequal probability sampling plan.
\newblock {\em Biometrika}, 69:653--656, 1982.

\bibitem{Clarkson86}
Kenneth~L. Clarkson.
\newblock Linear programming in $o(n 3^{d^2})$ time.
\newblock {\em Inf. Process. Lett.}, 22(1):21--24, 1986.

\bibitem{Clarkson95}
Kenneth~L. Clarkson.
\newblock Las vegas algorithms for linear and integer programming when the
  dimension is small.
\newblock {\em J. {ACM}}, 42(2):488--499, 1995.

\bibitem{ClarksonS89}
Kenneth~L. Clarkson and Peter~W. Shor.
\newblock Application of random sampling in computational geometry, {II}.
\newblock {\em Discrete {\&} Computational Geometry}, 4:387--421, 1989.

\bibitem{ITbook}
Thomas~M. Cover and Joy~A. Thomas.
\newblock {\em Elements of information theory {(2.} ed.)}.
\newblock Wiley, 2006.

\bibitem{CB99}
David~J. Crisp and Christopher J.~C. Burges.
\newblock A geometric interpretation of v-svm classifiers.
\newblock In {\em NIPS}, pages 244--250, 1999.

\bibitem{Dyer86}
Martin~E. Dyer.
\newblock On a multidimensional search technique and its application to the
  euclidean one-centre problem.
\newblock {\em {SIAM} J. Comput.}, 15(3):725--738, 1986.

\bibitem{DF89}
Martin~E. Dyer and Alan~M. Frieze.
\newblock A randomized algorithm for fixed-dimensional linear programming.
\newblock {\em Math. Program.}, 44(1-3):203--212, 1989.

\bibitem{GJ09}
Bernd G{\"{a}}rtner and Martin Jaggi.
\newblock Coresets for polytope distance.
\newblock In {\em SOCG}, pages 33--42, 2009.

\bibitem{GSZ11}
Michael~T. Goodrich, Nodari Sitchinava, and Qin Zhang.
\newblock Sorting, searching, and simulation in the mapreduce framework.
\newblock In {\em ISAAC}, pages 374--383, 2011.

\bibitem{GM08}
Sudipto Guha and Andrew McGregor.
\newblock Tight lower bounds for multi-pass stream computation via pass
  elimination.
\newblock In {\em ICALP}, pages 760--772, 2008.

\bibitem{HW87}
David Haussler and Emo Welzl.
\newblock epsilon-nets and simplex range queries.
\newblock {\em Discrete {\&} Computational Geometry}, 2:127--151, 1987.

\bibitem{DPSV12}
Hal~Daum{\'{e}} III, Jeff~M. Phillips, Avishek Saha, and Suresh
  Venkatasubramanian.
\newblock Efficient protocols for distributed classification and optimization.
\newblock In {\em ALT}, pages 154--168, 2012.

\bibitem{IMRUVY17}
Piotr Indyk, Sepideh Mahabadi, Ronitt Rubinfeld, Jonathan Ullman, Ali Vakilian,
  and Anak Yodpinyanee.
\newblock Fractional set cover in the streaming model.
\newblock In {\em {APPROX/RANDOM} 2017}, pages 12:1--12:20, 2017.

\bibitem{Kalai92}
Gil Kalai.
\newblock A subexponential randomized simplex algorithm (extended abstract).
\newblock In {\em STOC}, pages 475--482, 1992.

\bibitem{KSV10}
Howard~J. Karloff, Siddharth Suri, and Sergei Vassilvitskii.
\newblock A model of computation for mapreduce.
\newblock In {\em SODA}, pages 938--948, 2010.

\bibitem{LS14}
Yin~Tat Lee and Aaron Sidford.
\newblock Path finding methods for linear programming: Solving linear programs
  in {\~{o}}(vrank) iterations and faster algorithms for maximum flow.
\newblock In {\em FOCS}, pages 424--433, 2014.

\bibitem{MVPV17}
{Makrynioti, Nantia and Vasiloglou, Nikolaos and Pasalic, Emir and Vassalos,
  Vasilis}.
\newblock {Data Science with Linear Programming}.
\newblock
  \url{http://delbp.github.io/DeLBP-2017/papers/DeLBP-2017_paper_2CR.pdf},
  2017.

\bibitem{MSW96}
Ji{\v{r}}{\'{\i}} Matou{\v{s}}ek, Micha Sharir, and Emo Welzl.
\newblock A subexponential bound for linear programming.
\newblock {\em Algorithmica}, 16(4/5):498--516, 1996.

\bibitem{McGregor18}
Andrew McGregor.
\newblock private communication.

\bibitem{Megiddo84}
Nimrod Megiddo.
\newblock Linear programming in linear time when the dimension is fixed.
\newblock {\em J. {ACM}}, 31(1):114--127, 1984.

\bibitem{MiltersenNSW98}
Peter~Bro Miltersen, Noam Nisan, Shmuel Safra, and Avi Wigderson.
\newblock On data structures and asymmetric communication complexity.
\newblock {\em J. Comput. Syst. Sci.}, 57(1):37--49, 1998.

\bibitem{Mulmuley94}
Ketan Mulmuley.
\newblock {\em Computational geometry - an introduction through randomized
  algorithms}.
\newblock Prentice Hall, 1994.

\bibitem{MP80}
J.~Ian Munro and Mike Paterson.
\newblock Selection and sorting with limited storage.
\newblock {\em Theor. Comput. Sci.}, 12:315--323, 1980.

\bibitem{PVZ12}
Jeff~M. Phillips, Elad Verbin, and Qin Zhang.
\newblock Lower bounds for number-in-hand multiparty communication complexity,
  made easy.
\newblock In {\em SODA}, pages 486--501, 2012.

\bibitem{RoughgardenVW16}
Tim Roughgarden, Sergei Vassilvitskii, and Joshua~R. Wang.
\newblock Shuffles and circuits: (on lower bounds for modern parallel
  computation).
\newblock In {\em SPAA}, pages 1--12, 2016.

\bibitem{SenV08}
Pranab Sen and Srinivasan Venkatesh.
\newblock Lower bounds for predecessor searching in the cell probe model.
\newblock {\em J. Comput. Syst. Sci.}, 74(3):364--385, 2008.

\bibitem{Tao18}
Yufei Tao.
\newblock Massively parallel entity matching with linear classification in low
  dimensional space.
\newblock In {\em ICDT}, pages 20:1--20:19, 2018.

\bibitem{TKC05}
Ivor~W. Tsang, James~T. Kwok, and Pak{-}Ming Cheung.
\newblock Core vector machines: Fast {SVM} training on very large data sets.
\newblock {\em Journal of Machine Learning Research}, 6:363--392, 2005.

\bibitem{VC15}
Vladimir~N Vapnik and A~Ya Chervonenkis.
\newblock On the uniform convergence of relative frequencies of events to their
  probabilities.
\newblock In {\em Measures of complexity}, pages 11--30. Springer, 2015.

\bibitem{WD81}
Roberta~S. Wenocur and Richard~M. Dudley.
\newblock Some special vapnik-chervonenkis classes.
\newblock {\em Discrete Mathematics}, 33(3):313--318, 1981.

\bibitem{Yao79}
Andrew~Chi{-}Chih Yao.
\newblock Some complexity questions related to distributive computing
  (preliminary report).
\newblock In {\em STOC}, pages 209--213, 1979.

\bibitem{Yao83}
Andrew~Chi{-}Chih Yao.
\newblock Lower bounds by probabilistic arguments (extended abstract).
\newblock In {\em FOCS}, pages 420--428, 1983.

\bibitem{YT89}
Yinyu Ye and Edison Tse.
\newblock An extension of karmarkar's projective algorithm for convex quadratic
  programming.
\newblock {\em Math. Program.}, 44(1-3):157--179, 1989.

\end{thebibliography}
